\documentclass[letterpaper,english,reprint,aps,pra,superscriptaddress,longbibliography]{revtex4-2}

\usepackage{graphicx,mathrsfs,enumerate,epstopdf}
\usepackage[utf8]{inputenc}
\usepackage{mathtools}
\usepackage{amsfonts,amstext}
\usepackage{amsmath,amsthm,amssymb,tensor,diagbox,multirow}
\usepackage[hidelinks]{hyperref}\hypersetup{pdfpagemode=UseNone,pdfstartview=}
\usepackage[left=2cm,right=2cm]{geometry}
\usepackage{xcolor}
\usepackage{url}
\usepackage{times}
\usepackage{soul}
\usepackage{dsfont}
\usepackage[ruled,vlined,linesnumbered]{algorithm2e}

\makeatletter

\newtheorem{dfn}{Definition}
\newtheorem{thm}{Theorem}

\newtheorem{lem}[thm]{Lemma}

\newcommand{\bra}[1]{\langle #1|}
\newcommand{\ket}[1]{|#1\rangle}

\newcommand{\ketbra}[2]{| #1 \rangle \langle #2 |}

\usepackage[caption=false]{subfig}

\makeatother

\usepackage{babel}
\begin{document}
%\title{Quantum signal processing for Fourier approximations}
\title{{Fourier-based quantum signal processing}}
\author{Thais L. Silva}
\affiliation{Federal University of Rio de Janeiro, Caixa Postal 68528, Rio de Janeiro, RJ 21941-972, Brazil}
\affiliation{Quantum Research Centre, Technology Innovation Institute, Abu Dhabi,
UAE}

\author{Lucas Borges}
\affiliation{Federal University of Rio de Janeiro, Caixa Postal 68528, Rio de Janeiro, RJ 21941-972, Brazil}
\affiliation{Quantum Research Centre, Technology Innovation Institute, Abu Dhabi,
UAE}
% \author{M\'arcio M. Taddei}
% \affiliation{Federal University of Rio de Janeiro, Caixa Postal 68528, Rio de Janeiro, RJ 21941-972, Brazil}
% \affiliation{ICFO - Institut de Ciencies Fot\`oniques, The Barcelona Institute of Science and Technology, 08860, Castelldefels, Barcelona, Spain}
% \author{Stefano Carrazza}
% \affiliation{TIF  Lab,  Dipartimento  di  Fisica,  Universit\`a  degli  Studi  di  Milano  and  INFN  Sezione  di  Milano,  Milan, Italy}
% \affiliation{Quantum Research Centre, Technology Innovation Institute, Abu Dhabi, UAE} 
\author{Leandro Aolita}
\affiliation{Federal University of Rio de Janeiro, Caixa Postal 68528, Rio de Janeiro, RJ 21941-972, Brazil}
\affiliation{Quantum Research Centre, Technology Innovation Institute, Abu Dhabi,
UAE}

%\date{\today}
%%%%%%%%%%%%%%%%%%%%%%%%%%%%%%%%%%%%%%%%%%%%%%%%%%%%%%%%
\begin{abstract}
Implementing general functions of operators is a powerful tool in quantum computation. It can be used as the basis for a variety of quantum algorithms including matrix inversion, real and imaginary-time evolution, and matrix powers. Quantum signal processing is the state of the art for this aim, assuming that the operator to be transformed is given as a block of a unitary matrix acting on an enlarged Hilbert space. Here we present an algorithm for Hermitian-operator function design {from an oracle given by the unitary evolution with respect to that operator at a fixed time.} Our algorithm implements a Fourier approximation of the target function  based on the iteration of a basic sequence of single-qubit gates, for which we prove the expressibility. In addition, we present an efficient classical algorithm for calculating its parameters from the Fourier series coefficients. Our algorithm uses only one qubit ancilla regardless the degree of the approximating series. This contrasts with previous proposals, which required an ancillary register of size growing with the expansion degree. {Our methods are compatible with Trotterised Hamiltonian simulations schemes and hybrid digital-analog approaches.}
\end{abstract}
%%%%%%%%%%%%%%%%%%%%%%%%%%%%%%%%%%%%%%%%%%%%%%%%%%%%%%%%
\maketitle

%%%%%%%%%%%%%%%%%%%%%%%%%%%%%%%%%%%%%%%%%%%%%%%%%%%%%%%%
%%%%%%%%%%%%%%%%%%%%%%%%%%%%%%%%%%%%%%%%%%%%%%%%%%%%%%%%
\section{\label{sec:intro}Introduction}

Quantum signal processing (QSP) was originally developed as a technique to design real-variable functions from single-qubit rotations \cite{Low2016PRX}. It has attracted a lot of attention in the last few years after the discovery that it can be extended to realize also functions of operators, initially only for Hermitian operators \cite{low2017optimal,Low2019hamiltonian} and culminating in the general  formalism of quantum singular value transformations (QSVT) \cite{gilyen2019,Martyn2021GrandUni}.  Using single-qubit rotations and having controlled black-box access to the operator to be transformed, QSVT provides a circuit structure and error control that reduces the problem of operator-function synthesis to that of real-variable function approximation,
%depend only on the  function being implemented,
regardless the dimension of the underlying Hilbert space, the input operator, and the  state on which the operator function is to be applied. 

Basically, the $N$-qubits operator $H$   is  embedded in a unitary operator -- the oracle -- acting on an enlarged Hilbert space. A powerful model of oracle is the so called block-encoding oracle, which has the (necessarily normalized) non-unitary operator $H$ as a block \cite{Low2019hamiltonian,gilyen2019}. The QSVT circuit applying an approximation to a target operator  function $f[H]$ is obtained by interspersing the action of the oracle with qubit rotations on an ancilla control. The particular sequence of qubit gates in each call to the oracle allows for certain achievable functions, and the specific rotation angles determine the final function implemented. In the case of block-enconding oracles,  Chebyshev series approximations of the target function have been extensively studied. The (usually non-unitary) operator function is then obtained after post-selection on the control ancilla.

Several algorithms have been proposed using QSVT with block-encoding oracles for many different tasks such as Hamiltonian simulation \cite{low2017optimal}, imaginary time evolution \cite{silva2022fragmented}, matrix inversion \cite{gribling2021improving}, to give few examples \cite{gilyen2019}. Although very versatile and, in principle, possible to implement for any normalized operator (not necessarily a Hermitian or even square matrix), this kind of oracle requires a large number of ancillas, besides requiring for QSVT a further transformation called qubitization, which uses an extra ancilla and $\mathcal{O}(N)$ extra gates in each oracle call \cite{Low2019hamiltonian}.

 Here, we assume  $H$ to be Hermitian and explore an alternative type of oracle, namely the unitary evolution oracle, in which we have black-box access to the unitary $e^{-itH}$ with adjustable time $t$.  We apply the QSP ideas to implement operator functions via their truncated Fourier expansion. Previous algorithms were already proposed that use the time evolution given by a Hermitian operator to implement functions of that operator. Remarkable examples are the HHL algorithm for matrix inversion \cite{HHL2009} (and subsequent improvements of it \cite{ambainis2012,Childs2017}), and imaginary-time evolution for preparing thermal Gibbs states \cite{chowdhury2017}. Other algorithms were also put forward for more general smooth functions of Hermitian operators \cite{vanApeldoorn2020quantumsdpsolvers,shantanav2019}.  In common all these algorithms have the fact that they rely on the implementation of the unitary evolution operator for several values of time and use the technique of linear combination of unitaries \cite{Childs2012,Berry2015a} to obtain a Fourier-like sum from it. It requires  a number of ancillas that is logarithmic with the degree of the Fourier expansion. Our algorithm is a novel QSP variant that demands Hamiltonian simulation with a fixed time. It is superior to previous Fourier-based approaches in that, remarkably, it makes use of only one ancillary qubit, which is the minimum number of qubits necessary to state-independently implement non-unitary operators. Moreover, this method  does not require qubitization \cite{Low2019hamiltonian}, which is a significant experimental simplification in view of intermediate scale implementations. The method formalizes and expands on a technique that we introduced in a  summarized fashion in Ref.\ \cite{silva2022fragmented}. There, we  applied it to the specific case of an exponential function.

Realizing the unitary $e^{-itH}$ {--} a problem known as Hamiltonian simulatio{n --} can be a hard task  in itself. In fact, it is a BQP-complete problem \cite{Osborne2012}. Nevertheless, one may for instance apply product
formulae \cite{lloyd1996universal,campbell2019,COS19} to implement the oracle with gate complexities
 that, for intermediate-scale systems, can be considerably
smaller than for block-encoding oracles \cite{silva2022fragmented}. Furthermore, the real-time evolution oracle naturally arises in  hybrid analogue-digital platforms \cite{Arrazola2016,Parra-Rodriguez2020,Gonzalez-Raya2021}, for which
QSP schemes have already been studied \cite{HamiltonianQSP21}.

The basic sequence of single-qubit gates we use to assemble the Fourier series has been proposed in a previous work \cite{PerezSalinas2021} in the context of single-qubit variational circuits for approximating real-variable functions. Here,  we formally prove that this sequence is able to implement any normalized Fourier series of a real variable. Part of the proof is a generalization of the results in Ref.\ \cite{Haah2019product}, removing the restriction over the series parity.  Moreover, we provide an efficient classical algorithm to analytically calculate the rotation angles from the target series. This removes the need  for optimizing a large number of parameters present in Ref.\ \cite{PerezSalinas2021}, a task that becomes impractical as the number of rotation angles grows with the order of the series. This single-qubit construction is the basis of our operator-valued function algorithm. The complexity of the algorithm is dominated by the truncation order of the Fourier series and by the probability of success -- the  correct operator-function is obtained after post-selecting the state of the ancilla register. On one hand, standard Fourier series present convergence issues for approximating non-periodic functions -- the so-called Gibbs phenomenon. On the other hand, methods used to improve the convergence rate might incur in a reduction of the success probability.  Here, we  compare two  techniques for obtaining  Fourier series of non-periodic functions. The first was  introduced in Ref.\ \cite{vanApeldoorn2020quantumsdpsolvers}, which obtains the Fourier series from  a power series of the target function. The second is an analytical extension of the target function \cite{Boyd2002}. While the former provides an analytical bound for the error of the approximation, the latter has the advantage of avoiding a decrease of the success probability.

The paper is organized as follows: in Sec.\ \ref{sec:Preliminaries} we give some preliminary definitions and formally establish the problem we will tackle. In Sec.\ \ref{sec:results} we present our results, starting with the construction for real-variable function approximation using qubit rotations in Sec.\ \ref{sec:realvar}, which is extended to operator function design in Sec.\ \ref{sec:QSP_RTE}.  In Secs.\ \ref{sec:boundedError} and \ref{sec:periodicExt} we discuss two alternatives for obtaining Fourier approximations of non-periodic functions and compare them for the case of imaginary-time evolution. The proofs of the theorems and lemmas are left to Sec.\ \ref{sec:proofs}. Finally, we discuss our results in Sec.\ \ref{sec:Discussions}. 

\section{Preliminaries}\label{sec:Preliminaries}

From calls to a block-encoding oracle for a multi-qubit Hermitian operator $H$, QSP yields an operator which has the same eigenvectors as $H$, but eigenvalues transformed by a polynomial function.   This requires the use of qubit ancillas on which a sequence of  rotations is applied. The basic QSP circuit has the oracle as an input and the function being realized solely depends on the values of the parameters for the ancilla rotations. The achievable functions can be determined by analysing an $SU(2)$ operator as a function of a single real variable. By using the ancillary system to control the action of the operator oracle, it is possible to promote a real-variable function $f$ to an operator function $f[H]$ by applying $f$ to each eigenvalue of $H$. QSVT generalizes this procedure to include non-Hermitian operators $H$. In this case, eigenvalues and eigenvectors should be substituted by singular values and singular vectors.

We consider an $N$-qubit system $\mathcal{S}$, with Hilbert space $\mathbb H_\mathcal{S}$. A Hamiltonian operator $H$ on $\mathbb H_\mathcal{S}$ with eigenvalues $\lambda\in[\lambda_{\rm{min}},\lambda_{\rm{max}}]$ is given by a unitary oracle $O$. For simplicity, we assume $\lambda_{\rm{min}}\geq-1$ and $\lambda_{\rm{max}}\leq1$, such that $\|H\|\leq 1$. Contrary to other oracle types, for which the normalization is mandatory \cite{Low2019hamiltonian}, here it is merely a convenience. Our goal is to approximately implement an operator function $f[H]=\sum_\lambda f(\lambda) \ketbra{\lambda}{\lambda}$ on $\mathbb{H}_{\mathcal{S}}$ from calls to $O$, where $f:\mathbb{R}\mapsto \mathbb{C}$ and $\ket{\lambda}$ is the eigenvector of $H$ corresponding to eigenvalue $\lambda$. In general, $f[H]$ is not unitary and its implementation via a unitary circuit is achieved by enlarging the Hilbert space to include an ancillary register $\mathcal{A}$, whose Hilbert space we denote by $\mathbb H_\mathcal{A}$. 
We denote by $\mathbb H_\mathcal{SA}:= \mathbb H_\mathcal{S}\otimes\mathbb H_\mathcal{A}$ the joint Hilbert space of $\mathcal{S}$ and $\mathcal{A}$, and by $\left\Vert A \right\Vert$ the spectral norm of an operator $A$.  The oracle model that we consider is formally defined below. It encodes $H$ through the real-time 
unitary evolution it generates for a fixed but tunable time value. 
%%%%%%%%%%%%%%%%%%%%%%%%%%%%%%%%%%%%%%%%%%%%%%%%%%%%%%%%
\begin{dfn}
\label{def:real_t_or} (Real-time evolution Hamiltonian oracle). We refer as a real-time evolution oracle for a Hamiltonian $H$ on $\mathbb H_\mathcal{S}$ at a time $t\in\mathbb{R}$ to a controlled-$e^{-itH}$ gate  {$O=\mathds{1}\otimes\ket{0}\bra{0}+e^{-itH}\otimes\ket{1}\bra{1}$} on $\mathbb H_\mathcal{SA}$.
\end{dfn}

 The target operator function is  obtained via post-selection on the ancilla state after aplying a unitary operator $U_{f[H]}$ on $\mathbb H_\mathcal{SA}$ that encodes $f[H]$ in one of its matrix blocks {\cite{Low2019hamiltonian, gilyen2019, shantanav2019}}. 
$U_{f[H]}$ is called a block-enconding of $f[H]$, as defined below for a general operator $A$ on $\mathbb H_\mathcal{S}$ \cite{silva2022fragmented}:
%%%%%%%%%%%%%%%%%%%%%%%%%%%%%%%%%%%%%%%%%%%%%%%%%%%%%%%%
\begin{dfn} (Block encodings). For sub-normalization $0\leq\alpha\leq 1$ and tolerated error $\varepsilon> 0$, a unitary operator $U_A$ on $\mathbb H_\mathcal{SA}$
is an $(\varepsilon,\alpha)$-block-encoding of a linear operator $A$ on $\mathbb H_\mathcal{S}$ if $\left\Vert \alpha\,A-\bra{0}\,U_A\,\ket{0}\right\Vert \leq\,\varepsilon$,
for some $\ket{0}\in\mathbb{H}_\mathcal{A}$. 
For $\varepsilon=0$ and $(\varepsilon,\alpha)=(0,1)$ we use the short-hand terms perfect $\alpha$-block-encoding  and perfect block-encoding, respectively. 
\label{def:block_enc}
\end{dfn} 

The need for a subnormalization comes from the unitarity of $U_A$, whose blocks must therefore necessarily have spectral norm below unit. If the system and the ancilla are prepared in the joint state $\ket{\Psi}\ket{0}$, then the target operator $A$ on $\mathbb H_\mathcal{S}$ is obtained after measuring the ancilla whenever its state is projected onto $\ket{0}$ successfully. This happens with a success probability given by $p_\Psi(\alpha A)=\|\alpha A\ket{\Psi}\|^2$. Therefore, the larger the subnormalization, the best for the algorithm, since it will have a higher probability of getting a successful run. 

In the usual QSP method, the target function $f(\lambda)$ is approximated by a finite-order power series $\tilde{f}(\lambda)$ such that $|f(\lambda)-\tilde{f}(\lambda)|\leq \varepsilon$ for all $\lambda\in[\lambda_{\rm{min}},\lambda_{\rm{max}}]$. This ensures that the operator-function spectral-norm error is bounded by the same $\varepsilon$. Similarly, here we use a truncated Fourier series to approximate the target function. The standard measure of complexity used in oracle based algorithms is the number of calls to the oracle needed to implement an approximation of $f[H]$ up to spectral error $\varepsilon$, which we denote by  $q_f(\varepsilon)$. As we show later, $q_f(\varepsilon)$ equals double the truncation order of the approximating series.  In the next section, we show how it is possible to generate an arbitrary Fourier series from single-qubit rotations. The exact same construction can be used to implement functions at the level of the eigenvalues of an operator encoded in the oracle of Def.\ \ref{def:real_t_or}. This is due to the fact that  each eigenvalue $\lambda$ of $H$ is naturally attributed a two-dimensional subspace spanned by $\{\ket{\lambda}\ket{0},\ket{\lambda}\ket{1}\}$.

\section{Results}
\label{sec:results}

\subsection{Single-qubit rotation synthesis of Fourier series}
\label{sec:realvar}

Consider a Fourier series $\tilde{g}_q(x)=\sum_{m=-q/2}^{q/2} c_m\, e^{im x}$ such that $|g(x)-\tilde{g}_q(x)|\leq \varepsilon_0$ for all $x\in[-\pi,\pi]$, where $g:\mathbb{R}\rightarrow\mathbb{C}$ is such that $|g(x)|\leq 1$ on the interval $[-\pi,\pi]$. The function $g$ will serve to approximate the target $f$ in a reduced interval as to avoid the so called Gibbs phenomenon, as we explain later. For now, the important thing is that we would like to build a unitary single-qubit operator
\begin{equation}\label{eq:single_U}
 U_{\tilde{g}_q\tilde{h}_q}(x) = \left(\begin{array}{cc}
    \tilde{g}_q(x)&\tilde{h}_q(x)\\
    -\tilde{h}^*_q(x)&\tilde{g}^*_q(x)
   \end{array}\right)
\end{equation}
having $\tilde{g}_q(x)$ as one of its entries with an arbitrary Fourier series  $\tilde{h}_q(x)$ of order $q/2$ as the complementary entry. The operator should be obtained by repeatedly applying a basic gate $R(x,\boldsymbol{\xi})$ with the variable $x$ as input, such that,  for a convenient choice of parameters $\boldsymbol{\xi}_k$,  $U_{\tilde{g}_q\tilde{h}_q}(x) = \prod_{k=0}^{q}R(x,\boldsymbol{\xi}_k)$ for any $x\in [-\pi,\pi]$ given as input.

Inspired by a construction in Ref.\ \cite{PerezSalinas2021}, here we consider the basic single-qubit gate  $R(x,\omega,\boldsymbol{\xi})=e^{i\frac{\zeta+\eta}{2} Z}e^{-i\varphi Y}e^{i\frac{\zeta-\eta}{2} Z}e^{i\omega x Z}e^{-i\kappa Y}$, which has five adjustable parameters $\{\omega,\boldsymbol{\xi}\}\in\mathbb{R}^5$, where $\boldsymbol{\xi}:=\{\zeta,\eta,\varphi,\kappa\}$.  $X$, $Y$ and $Z$ are the Pauli operators. The input variable $x\in\mathbb{R}$ is the signal to be processed and $e^{i \omega x Z}$ is called iterate. In Ref. \cite{PerezSalinas2021}, it was observed that the gate sequence $\mathcal{R}(x,\boldsymbol{\omega},\boldsymbol{\Phi}):=\prod_{k=0}^{q}R(x,\omega_k,\boldsymbol{\xi}_k)$, with $\boldsymbol{\omega}:=\{\omega_0,\cdots, \omega_q\}\in\mathbb{R}^{q+1}$ and $\boldsymbol{\Phi}=\{\boldsymbol{\xi}_0,\cdots,\boldsymbol{\xi}_{q}\}\in\mathbb{R}^{4(q+1)}$, can encode certain finite Fourier series into its matrix components. There, the authors numerically find the sequence of pulses for some examples of real functions. However, they do not prove the existence of a complementary Fourier series $\tilde{h}_q(x)$, nor provide an analytical mean of calculating the pulses, relying on numerical optimizations of the rotation parameters and thus limiting the order of the actual implementable series. Here, not only do we formally prove that $U_{\tilde{g}_q\tilde{h}_q}(x)=\mathcal{R}(x,\boldsymbol{\omega},\boldsymbol{\Phi})$ can encode any target series but also we provide an explicit, efficient recipe for finding the adequate choice of pulses $\boldsymbol{\Phi}$. This is the content of the following theorem, whose proof can be found in Sec.\ \ref{sec:proofRV}. A circuit representation of $\mathcal{R}(x,\boldsymbol{\omega},\boldsymbol{\Phi})$ can be found in Fig.\ \ref{fig:generalf}.a).
\begin{thm}\label{lem:qsp2}[Single-qubit Fourier series synthesis]
 Given $\tilde{g}_q(x)=\sum_{m=-q/2}^{q/2} c_m\, e^{im x}$, with $q\in\mathbb{N}$ even, there exist $\boldsymbol{\omega}$ and  $\boldsymbol{\Phi}$ such that $\bra{0}\,\mathcal{R}(x,\boldsymbol{\omega},\boldsymbol{\Phi})\,\ket{0}=\tilde{g}_q(x)$ for all $|x|\leq \pi$ iff $|\tilde{g}_q(x)|\leq 1$ for all $|x|\leq\pi$. Moreover, $\boldsymbol{\omega}$ can be taken such that  $\omega_0=0$ and $\omega_k=(-1)^k/2$, for all $1\leq k\leq q$, and $\boldsymbol{\Phi}$ can be calculated classically from $\{c_m\}_{m}$ in time $\mathcal{O}\left(\text{poly}(q)\right)$.
\end{thm}
%%%%%%%%%%%%%%%%%%%%%%%%%%%%%%%%%%%%%%%%%%%%%%%%%%%%%%%%%%%

The power of Theo.\ \ref{lem:qsp2} relies on that it allows for the implementation of any complex Fourier series, assuming only that it is properly normalized. Moreover, as is usual in QSP, the number of pulses is determined solely by the order of the implemented series. Now, as a Fourier series can be used to approximate any complex function $g(x)$ on a finite interval, we have a way to approximately obtain any complex function from qubit rotations. Moreover, the error of the approximant is exactly the error made in the series approximation, that is $\varepsilon=\max_{x\in[-\pi,\pi]}|g(x)-\tilde{g}_q(x)|$. Notice that if it is guaranteed that $|g(x)|\leq 1$, then at most $|\tilde{g}_q(x)|\leq1+\varepsilon$ and the normalized series $\tilde{g}_q(x)/(1+\varepsilon)$ satisfies Theo.\ \ref{lem:qsp2} with error of at most $2\varepsilon$.

The proof presented in Sec.\ \ref{sec:proofRV} follows a similar strategy to that of Refs.\ \cite{Haah2019product,chao2020finding}. The first step   is to prove the existence of $\tilde{h}_q(x)$, which leads to a method to obtain it from the coefficients of $\tilde{g}_q(x)$.   It is easy to see that a unitary operator can always be built from $\tilde{g}_q(x)$ by taking $\tilde{h}_q(x)=\sqrt{1-|\tilde{g}_q(x)|^2}$. However, this choice is not unique and it is not obvious that $\tilde{h}_q(x)$ could also be chosen as a Fourier series with the same order as $\tilde{g}_q(x)$, neither it is obvious how to obtain this complement. The next step  is to show that indeed the sequence of pulses produce all the frequencies of the Fourier series. Finally, this two elements are combined in  a constructive method to show that $U$ in Eq.\ \eqref{eq:single_U} can be obtained as the sequence of gates $\mathcal{R}(x,\boldsymbol{\omega},\boldsymbol{\Phi})$. It directly furnishes a way to calculate $\boldsymbol{\Phi}$ from $\tilde{g}_q(x)$.

%%%%%%%%%%%%%%%%%%%%%%%%%%%%%%%%%%%%%%%%%%%%%%%%%%%%%%%%%%%%%

\subsection{Operator function design}
\label{sec:QSP_RTE}

Here, we elaborate on the technique introduced in Sec. V-C of Ref.\ \cite{silva2022fragmented}. We synthesize an ($\varepsilon,\alpha$)-block-encoding of $f[H]$ from an oracle for $H$ as in Def. \ref{def:real_t_or}. We build  a circuit $\mathcal{C}$ generating a perfect block-encoding $\boldsymbol{V}_{\boldsymbol{\Phi}}$ of a target Fourier expansion $\tilde{g}_q[Ht]:=\sum_{m=-q/2}^{q/2} c_m e^{imHt}$ that $\varepsilon$-approximates $\alpha\, f[H]$, for some $0<\varepsilon<1$,        $\alpha\leq 1$, and a suitable $t>0$. This is done by adjusting $\boldsymbol{\Phi}$ according to Theorem \ref{lem:qsp2}. Similarly, an algorithm for Fourier synthesis could be devised from the basic single-qubit gate $R(x,\phi)=e^{ixX}e^{i\phi Z}$, the usual pulse sequence used in QSP \cite{Low2016PRX}. However,  it would require one extra qubit ancilla and decrease the subnormalization by half (meaning a decrement of the success probability by $1/4$), as we show in App.\ \ref{app:proofFourier}. 

The function $\tilde{g}_q$ is taken as a Fourier approximation of an intermediary function $g$ such that $|g(x)-\tilde{g}_q(x)|\leq\varepsilon_0$ $\forall x\in[-\pi,\pi]$, and $|g(x_{\lambda})-\alpha\, f(\lambda)|\leq \varepsilon_1$ for all $x_{\lambda}:=\lambda\, t\in[-x_0,x_0]$, with  $x_0\leq \pi$ and $\varepsilon_0+\varepsilon_1\leq\varepsilon$. Since we are interested in approximating $f(\lambda)$ for all $\lambda\in[-1,1]$, we can take $t=x_0$ so that $x_{\lambda}=\lambda\, t$ is in $[-x_0,x_0]$ for all the eigenvalues of $H$. The function $g$ is, up to $\varepsilon_1$-approximation, a periodic extension of $f$.
%the interval of convergence of $\tilde{g}_q$ to $g$ for all $\lambda\in[-1,1]$. 
The reason for this intermediary step here is to circumvent the well-known Gibbs phenomenon, by virtue of which convergence of a Fourier expansion cannot in general be guaranteed at the boundaries. In turn, the sub-normalization factor $\alpha$ might take a value strictly smaller than $1$ even if $|f(\lambda)|\leq 1$ in $[-1,1]$ because our $\tilde{g}_q$ converges to $\alpha f$ only for $|x_{\lambda}|\leq x_0$, whereas Theorem \ref{lem:qsp2} requires that $|\tilde{g}_q(x_{\lambda})|\leq 1$ for all $|x_{\lambda}|\leq \pi$. This forces one to sub-normalize the expansion so as to guarantee normalization over the entire domain. (The inoffensive sub-normalization factor $(1+\varepsilon)^{-1}$ is neglected.)

%%%%%%%%%%%%%%%%%%%%%%%%%%%%%%%%%%%%%%%%%%%%%%%%%%%%%%%%%
\begin{figure}[t!]
\centering
\includegraphics[width=1\columnwidth]{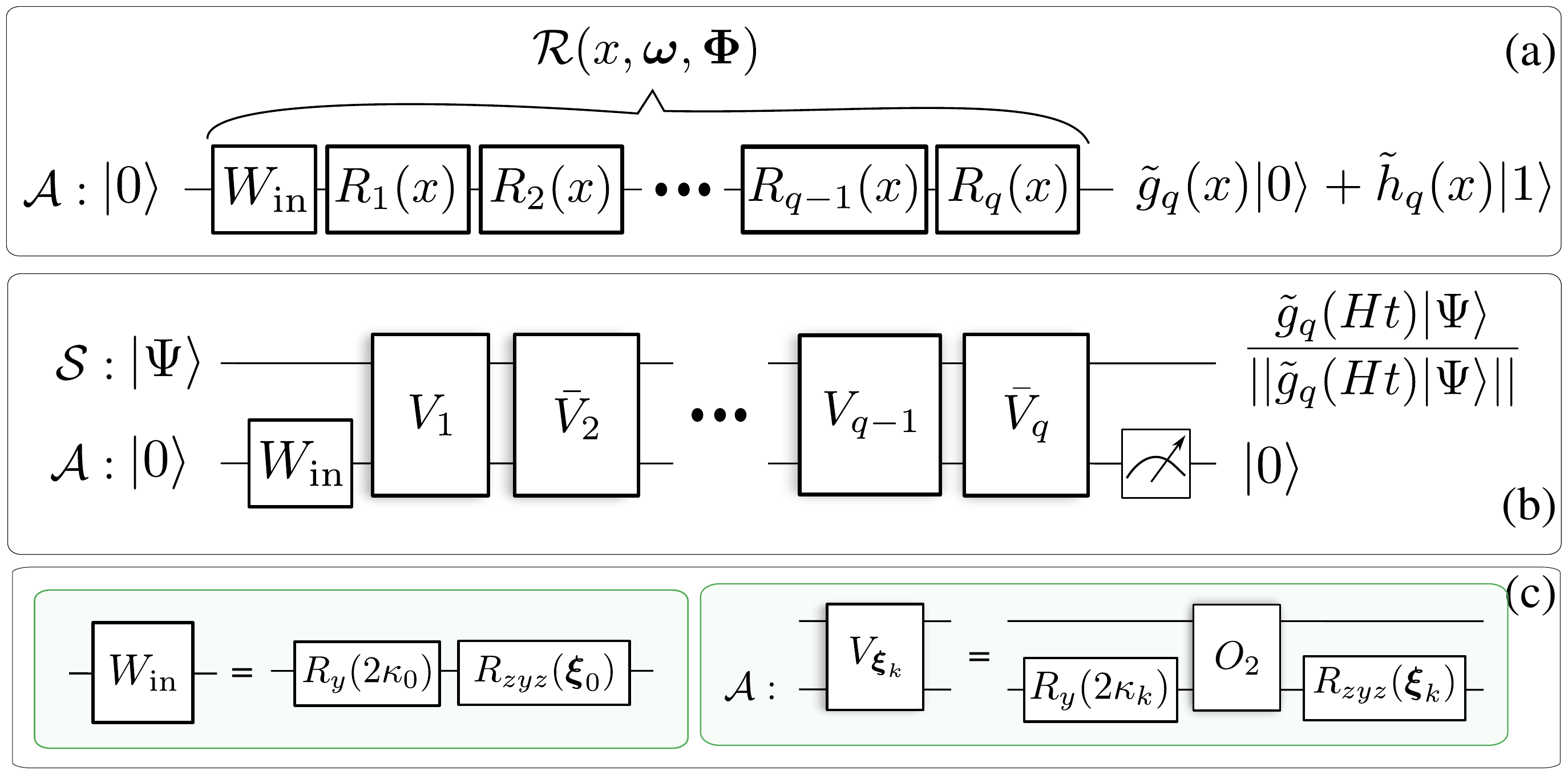}
\caption{\label{fig:generalf} \textbf{Circuit for generic operator-function design.} a) Single-qubit circuit implementing the unitary transformation $\mathcal{R}(x,\boldsymbol{\omega},\boldsymbol{\Phi})$ which contains $\tilde{g}_q(x)$ as a block. Here we use the short-hand notation $R_k(x):=R(x,\omega_k,\xi_k)$. b) The circuit  $\mathcal{C}$ from Alg. \ref{alg:frealtime} uses the register $\mathcal{A}$ as an ancilla. If the ancilla is initialised and post-selected in $\ket{0}$, the circuit prepares the system state $\frac{\tilde{g}_q[Ht]\ket{\Psi}}{||\tilde{g}_q[Ht]\ket{\Psi}||}$, which $\varepsilon$-approximates the target output $\frac{f[H]\ket{\Psi}}{\|f[H]\ket{\Psi}\|}$. 
The detailed components of  $\mathcal{C}$ are  shown in part c). $W_{\text{in}}$ is a fixed ancillary unitary. The basic blocks $V_k$ in panel b) represent the gates $V_{\boldsymbol{\xi}_k}$ in c).
Each $V_{\boldsymbol{\xi}_k}$ involves one query to the oracle $O$. $\bar{V}_k$ is defined as $V_k$ but with  $O^{\dagger}$ substituting $O$. Hence, the query complexity of $\mathcal{C}$ is  $q$. The approximating function $\tilde{g}_q$ is determined by the rotation angles $\boldsymbol{\Phi}=\{\boldsymbol{\xi}_0,\cdots,\boldsymbol{\xi}_{q}\}$. $R_\alpha(\phi)$ with $\alpha=x,y,z$ denotes a rotation in the $\alpha$ direction and $R_{zyz}(\boldsymbol{\xi}_k):=R_z(\zeta_k+\eta_k)\,R_y(2\varphi_k)\,R_z(\zeta_k-\eta_k)$, with $\boldsymbol{\xi}_k=\{\zeta_k,\eta_k,\varphi_k,\kappa_k\}$. These angles are chosen such that $\tilde{g}_q$ is a high-precision Fourier approximation of  $f$ in the interval $[-1,1]$.
}
 \end{figure}
%%%%%%%%%%%%%%%%%%%%%%%%%%%%%%%%%%%%%%%%%%%%%%%%%%%%%%%%%

Next, we explicitly show how to generate $\boldsymbol{V}_{\boldsymbol{\Phi}}$. We define  $\boldsymbol{V}_{{\boldsymbol{\Phi}}}:=\left(
 \bar{V}_{\boldsymbol{\xi}_{q}}{V}_{\boldsymbol{\xi}_{q-1}}\cdots \bar{V}_{\boldsymbol{\xi}_{2}}{V}_{\boldsymbol{\xi}_{1}} \right)W_{\text{in}}$, where the basic QSP blocks are given as
\begin{subequations}
 \begin{equation}
 \label{eq:vphiRT}
  {V}_{\boldsymbol{\xi}_k}:=\left[\mathds{1}\otimes\left(e^{i\frac{\zeta_k+\eta_k}{2} Z}e^{-i\varphi_k Y}e^{i\frac{\zeta_k-\eta_k}{2} Z}\right)\right]O \left[\mathds{1}\otimes e^{-i\kappa_k Y}\right]
 \end{equation}
 and 
 \begin{equation}\label{eq:vphibarRT}
  \bar{V}_{\boldsymbol{\xi}_k}:=\left[\mathds{1}\otimes\left(e^{i\frac{\zeta_k+\eta_k}{2} Z}e^{-i\varphi_k Y}e^{i\frac{\zeta_k-\eta_k}{2} Z}\right)\right]O^\dagger \left[\mathds{1}\otimes e^{-i\kappa_k Y}\right],
 \end{equation}
\end{subequations}
with ${\boldsymbol{\xi}_k}:=\{\zeta_k,\eta_k,\varphi_k,\kappa_k\}$. 
$V_{\boldsymbol{\xi}_k}$ and $\bar{V}_{\boldsymbol{\xi}_k}$ play a similar role to $R(x,\omega_k,\boldsymbol{\xi}_k)$ in Sec.\,\ref{sec:realvar}. The iterate is taken as the oracle: $O=\mathds{1}\otimes\ket{0}\bra{0}+ e^{-iHt}\otimes \ket{1}\bra{1}$ (with {$x_\lambda$} inside $O$ playing the role of $x$ in Sec.\,\ref{sec:realvar} for each $\lambda$).
Notice that $O$ readily acts as an $SU(2)$ rotation on each 2-dimensional subspace $\text{span}\{\ket{\lambda}\ket{0},\ket{\lambda}\ket{1}\}$. Here, qubitization \cite{Low2019hamiltonian} is not required.  We take
\begin{equation}
\label{eq:WinRT}
W_{\text{in}}=\mathds{1}\otimes\big[e^{i\frac{\zeta_0+\eta_0}{2} Z}e^{-i\varphi_0 Y}e^{i\frac{\zeta_0-\eta_0}{2} Z}e^{-i\kappa_0 Y}\big].
\end{equation}
The following pseudocode summarizes the entire procedure to implement $\boldsymbol{V}_{\boldsymbol{\Phi}}$.

\begin{center}
{\setlength{\fboxsep}{1pt}

%\framebox{
\begin{minipage}[t]{0.9\columnwidth}
\centering
\begin{algorithm}[H]\label{alg:frealtime}
\SetAlgoLined
\SetAlgorithmName{Algorithm}{Algorithm}{Algorithm}
\caption{Operator-valued Fourier series from real-time evolution Hamiltonian oracles}
\SetKwInOut{Input}{input}\SetKwInOut{Output}{output}
\SetKwData{Even}{even}

\Input{Fourier coefficients $\boldsymbol{c}:=\{c_m\}_{|m|\leq q/2}$, oracle $O$ for $H$ at a time $t=x_0$, and its inverse $O^{\dagger}$.}
\Output{unitary quantum circuit $\mathcal{C}$.} 
\BlankLine
calculate the rotation angles $\boldsymbol{\Phi}$ (Sec.\ \ref{sec:proofRV})\;
\Begin(construction of $\mathcal{C}$:){
apply $W_{\text{in}}$ from Eq.\ \eqref{eq:WinRT} on $\mathcal{A}$\;
\For{$k= 1$ \KwTo $k=q$}{
   \lIf{$k$ is odd}{apply $V_{\boldsymbol{\xi}_k}$ from Eq.\ \eqref{eq:vphiRT}} %
   \lElse{appy $\bar{V}_{\boldsymbol{\xi}_k}$ from Eq.\ \eqref{eq:vphibarRT}}
}
}
\end{algorithm}
\end{minipage}
}
\end{center}

For each eigenvalue $\lambda$ of  $H$, the operator $\boldsymbol{V}_{{\boldsymbol{\Phi}}}$ realizes  the rotation operator $ \mathcal{R}(x_\lambda,\boldsymbol{\omega},\boldsymbol{\Phi})$  on the ancillary register $\mathcal{A}$, as in Theorem \ref{lem:qsp2}. Therefore, the choice of $\boldsymbol{\Phi}$  allows us to block-encode any normalized Fourier series of $H$. This is the content of the next theorem. The circuit realizing $\boldsymbol{V}_{{\boldsymbol{\Phi}}}$ is depicted in Figs. \ref{fig:generalf}.b) and \ref{fig:generalf}.c).

\begin{thm}\label{fourier} (Operator-valued Fourier series from real-time evolution oracles) Let $\tilde{g}_q:[-\pi,\pi]\rightarrow\mathbb{C}$, with  $\tilde{g}_q(x)=\sum_{m=-q/2}^{q/2} {c}_m e^{i m x}$, be such that $|\tilde{g}_q(x)|\leq1$ for all $x\in[-\pi,\pi]$. Then there is a pulse sequence ${\boldsymbol{\Phi}}=\{ {\boldsymbol{\xi}}_0, \cdots,{\boldsymbol{\xi}}_{q} \}\in\mathbb{R}^{4(q+1)}$, with ${\boldsymbol{\xi}}_k=\{\zeta_k,\eta_k,\varphi_k,\kappa_k\}$, such that the operator $\boldsymbol{V}_{{\boldsymbol{\Phi}}}$ on $\mathbb{H}_{\mathcal{SA}}$ implemented by Alg. \ref{alg:frealtime} is a perfect block-encoding of $\tilde{g}_q[Ht]$, i.{e.} 
\begin{equation}
\bra{0}\boldsymbol{V}_{{\boldsymbol{\Phi}}}\ket{0}=\sum_{\lambda}\tilde{g}_q(x_\lambda)\,\ket{\lambda}\bra{\lambda},\label{eq:rtrealfunction}
\end{equation}
with $x_\lambda=\lambda t $. 
Moreover, the pulse sequence can be obtained classically in time $\mathcal{O}(\text{poly}(q))$.
\end{thm}

Let us further discuss  the convergence of $\tilde{g}_q$.
When $g$
is a periodic function with period $2\pi$, with a finite number of jump discontinuities, 
the standard partial Fourier series $\tilde{g}_q(x)=\sum_{m=-q/2}^{q/2}c_m\,e^{imx}$ with $c_m=\frac{1}{2\pi} \int_{-\pi}^{\pi}g(x)e^{-imx}dx$
converges to $g(x)$ when $q$ goes to infinity for all $x$ except for the discontinuities, where the series converges to the average value of the function at both sides of the discontinuity. Besides, the convergence close to the discontinuities is slow due to  the so-called Gibbs phenomenon. If $g$ is obtained by the simple periodic repetition of the values taken by $f$ -- usually a non-periodic continuous function -- inside the interval $[-\pi,\pi]$, then the borders of that interval present jump discontinuities.
In that case, $\tilde{g}_q(x)$ is a good approximation for
$g(x)$ only in the sense that, given  $\varepsilon>0$ and $\delta\in(0,\pi)$, there is a
large enough truncation order $q/2$ such that $\left|g(x)-\tilde{g}_q(x)\right|\leq\,\varepsilon$
for all $x\in[-\pi+\delta,\pi-\delta]$.  Moreover, this approach does not have a closed expression relating the error to $\delta$ and to the truncation order $q/2$, and a case-by-case analysis is required. This discussion justifies taking $g$ as equal to $\alpha f$ in a limited interval, as $f$ in general is not periodic. Outside $[-x_0,x_0]$, $g$ can take any values, provided it remains normalized and continuous. 
For instance, one could define $g$ as the periodic repetition of the continuous function $h:\mathbb{R}\rightarrow\mathbb{C}$ defined as
\begin{equation}\label{eq:g}
h(x)=
\begin{cases}
      h_1(x)\quad x\in[-\pi,-x_0]\\
      \alpha f(x/t) \quad x\in[-x_0,x_0]\\
      h_2(x) \quad x\in[x_0,\pi],
\end{cases}
\end{equation}
with $h_1$ and $h_2$ any continuous functions satisfying $h_2(\pi)=h_1(-\pi)$, $h_1(-x_0)=\alpha f(-x_0/t)$, and $h_2(x_0)=\alpha f(-x_0/t)$, such that $g$ itself is continuous.
Nevertheless, in  this case, possibly the first derivative of $g(x)$ in $x=-x_0$ and $x=x_0$ would present jump discontinuities, compromising the convergence of $\tilde{g}_q$. 

In what follows, we present two ways of obtaining a Fourier series that approximates the target function.  The first one  does not use the standard Fourier integral to obtain the series coefficients,  while the second method uses a filter to produce an analytical function $g$. While the first method has the advantage of providing analytical bounds for the complexity $q$ as a function of the error $\varepsilon$, the second one is classically much less expensive.

\subsection{Bounded-error Fourier series}
\label{sec:boundedError}

Consider that a power series $\tilde{f}_L(\lambda)=\sum_{l=0}^{L}a_l \lambda^l$ is given such that $|f(\lambda)-\tilde{f}_L(\lambda)|\leq\varepsilon/(4\alpha)$ for all $\lambda\in[-1,1]$. For the method presented in this section we do not need to specify $g$ outside the interval $[-x_0,x_0]$ and we take $\varepsilon_1=0$, such that 
\begin{equation}
\left|g(x_\lambda)-\alpha\sum_{l=0}^L \frac{a_l}{t^l}x_\lambda^l\right|\leq\varepsilon/4, 
\end{equation}
for all $x_\lambda\in[-x_0,x_0]$ with $t=x_0$.
We employ  a construction from Ref.\,\citep{vanApeldoorn2020quantumsdpsolvers} that, given $0<\delta<\pi/2$, $\alpha \tilde{f}_L(\lambda)$, and
\begin{equation}
\label{eq:q_rt}
 q\geq\left\lceil\frac{2\pi}{\delta}\ln\left(\frac{4}{\varepsilon}\right)\right\rceil
\end{equation}
 yields $\boldsymbol{c}=\{c_m\}_{|m|\leq q/2}$ such that $\tilde{g}_q$ $\varepsilon$-approximates $g$ for all $x_{\lambda}\in[-x_0,x_0]$ with $x_0=\pi/2-\delta$.  The important features of this construction are given in a slightly modified version of Lemma
37 of Ref. \citep{vanApeldoorn2020quantumsdpsolvers}:

\begin{lem}\label{tayfou} (Lemma 37 of Ref. \cite{vanApeldoorn2020quantumsdpsolvers}) Let $\delta\in(0,\pi/2)$,  $\varepsilon\in(0,1)$,
and $g:\mathbb{R}\rightarrow\mathbb{C}$ be such that $\left|g(x)-\sum_{l=0}^{L}d_{l}\left(\frac{2}{\pi}x\right)^{l}\right|\leq\,\varepsilon/4$
 for all $x\in[-\frac{\pi}{2}+\delta,\frac{\pi}{2}-\delta]$. Then $\exists\,\mathbf{c}\in\mathbb{C}^{q+1}$
such that 
\begin{equation}
\left|g(x)-\tilde{g}_q(x)\right|\leq\,\varepsilon\label{eq:newfourier}
\end{equation}
for all $x\in[-\frac{\pi}{2}+\delta,\frac{\pi}{2}-\delta]$, where $\tilde{g}_q(x)=\sum_{m=-q/2}^{q/2}c_{m}e^{i m x}$, $q=\text{max}\left[\left\lceil\frac{2\pi}{\delta}\ln\left(\frac{4\|d\|_{1}}{\varepsilon}\right)\right\rceil,0\right]$,
and $\|\mathbf{c}\|_{1}\leq\|\mathbf{d}\|_{1}$. Moreover, $\mathbf{c}$ can be efficiently calculated
on a classical computer in time $\textrm{poly}(L,q,log(1/\varepsilon))$.\end{lem}
 
For $f$ analytical, one can obtain the power series of $g$ from a truncated Taylor series of $f$ using that $g(x_{\lambda})=\alpha f(\lambda)$. The truncation order $L$ can be obtained from the remainder:
\begin{equation}
\label{eq:Taylor_remainder}
\frac{\varepsilon}{4}\leq\frac{\underset{\lambda\in[-1,1]}{\max}\left|\alpha\, f^{(L+1)}(\lambda)\right|}{(L+1)!}.
\end{equation}
 In addition,  we may bound the sub-normalization constant $\alpha$ in terms of the  Taylor coefficients $\boldsymbol{a}:=\{a_l\}_{0\leq l\leq L}$ of $f$ as follows. Lemma \ref{tayfou} can be applied to build  a Fourier series for $g(x_\lambda)$, identifying $d_l=\alpha a_l/(2t/\pi)^l= \alpha a_l/(1-\frac{2\delta}{\pi})^l$ since $t=x_0=\pi/2-\delta$. 
 The expansion $\tilde{g}_q(x_\lambda)$ obtained converges to the target function only in  $\left[-\frac{\pi}{2}+\delta,\frac{\pi}{2}-\delta\right]$,
 although the period of $\tilde{g}_q$ is still $2\pi$. To implement $\tilde{g}_q$ using Theorem \ref{fourier}, we need to chose $\alpha$ as to guarantee the normalization $|\tilde{g}_q(x_\lambda)|\leq1$ in the whole interval $[-\pi,\pi]$. For that, notice that $|\tilde{g}_q(x)|=\left|\sum_{m=-q/2}^{q/2}c_{m}e^{i m x}\right|\leq \sum_{m=-q/2}^{q/2}|c_{m}|=\|\mathbf{c}\|_1$ for all $x\in[-\pi,\pi]$. On the other hand, according to Lemma \ref{tayfou}, the vector of coefficients  satisfies $\|\mathbf{c}\|_{1}\leq\|\mathbf{d}\|_{1}$, ensuring that  $|\tilde{g}_q(x)|\leq 1$ is attained in the whole interval $[-\pi,\pi]$ if  $\|\mathbf{d}\|_{1}=\sum_{l=0}^L|d_l|=\alpha\sum_{l=0}^L\left|a_l/(1-\frac{2\delta}{\pi})^l\right|\leq1$. Therefore, it  suffices to take $\alpha$ such that 
\begin{equation}
\label{eq:sub_norm}
\sum_{l=0}^L\big|a_l/(1-2\,\delta/\pi)^l\big|\leq\alpha^{-1}.
\end{equation}
%The sum converges if $\lim_{l\rightarrow \infty}|\frac{a_{l+1}}{a_l}|<1-\frac{2\delta}{\pi}$. 
 Note that $L$ and $\alpha$ are inter-dependent. One way to determine them is to increase $L$ and iteratively adapt $\alpha$ until Eqs. \eqref{eq:Taylor_remainder} and \eqref{eq:sub_norm} are both satisfied. Alternatively, if the expansion converges sufficiently fast (e.g., if $\lim_{l\rightarrow \infty}|\frac{a_{l+1}}{a_l}|<1-\frac{2\delta}{\pi}$), one can simply substitute $L$ in Eq.\ \eqref{eq:sub_norm} by $\infty$, simplifying the analysis. 
Finally, the obtained $\boldsymbol{c}$ can be input to Alg. \ref{alg:frealtime} in order to produce the desired operator function.

We experimented applying Lemma \ref{tayfou} in the context of imaginary-time evolution  \cite{silva2022fragmented}. The desired function to be implemented there is  $f({\lambda}) = e^{-\beta(\lambda + 1)}$, where $\beta$ is the inverse temperature. While Lemma \ref{tayfou} does indeed provide a theoretically efficient recipe to obtain a Fourier-series approximation, and it also guarantees a controlled approximation error, it proved prohibitively time consuming to evaluate said coefficients numerically. This happens due to the $1/\delta$ dependence in the expansion order of the desired series - actually, the number of operations to obtain the coefficients from Lemma \ref{tayfou} scales as $\mathcal{O}(1/\delta ^{3})$. In essence, we want $\delta$ to be small so that the parameter $\alpha^{-1}$ in Eq. \eqref{eq:sub_norm} is also small, which will have implication in the success probability of the block-encoding implementation. For imaginary-time evolution, we \cite{silva2022fragmented} obtain $\delta = \mathcal{O}(1/\beta)$ for large $\beta$, as the $\delta$ that optimizes the overall number of calls to the oracle, so that the truncation order will increase linearly with $\beta$ and the computational complexity as $\mathcal{O}(\beta^{3})$. Since one is usually interested in the low temperature regime, the construction turned out to be too expensive for large $\beta$. An alternative to this construction is provided below.

\subsection{Analytic-extension Fourier series}
\label{sec:periodicExt}

It is known \cite{Webb2018} that the Fourier series of an analytical, periodic function converges exponentially fast, i.e, the error $\varepsilon$ made in the approximation goes asymptotically as $\varepsilon=\mathcal{e^{-cq}}$, $c > 0$,  where $q$ is the truncation order. Our goal is to obtain a periodic extension of $f$ which is analytic in the whole interval $[-\pi,\pi]$. The advantage of this method is that there will be no unnecessary sub-normalization. This would correspond to $\delta = 0$ in Lemma\ \ref{tayfou}, which is only possible there if we consider the exact, infinite series. For simplicity, we consider $t=1$ such that $x_\lambda=\lambda t=\lambda$ and, in a slight abuse of notation, we use the variable $\lambda$ everywhere, even  outside the interval $[-1,1]$ of eigenvalues of $H$. Suppose $f(\lambda)$ is analytic for $\lambda \in [-\pi, \pi]$ and $|f(\lambda)| \leq 1$ for $\lambda \in [-1, 1]$, then the function $g(\lambda):=\tilde{g}(\lambda, L, \chi) = f(\lambda)b(\lambda, L, \chi)$, with 
\begin{equation}
\label{eq:func_filter}
b(\lambda, L, \chi) = \frac{erf[L(\lambda + \chi)] - erf[L(\lambda - \chi)]}{2},
\end{equation}
where $erf$ is the error function, $1<\chi<\pi $, is an approximate, analytic \cite{Boyd2002} extension of $f$ in the interval  $\lambda \in [-\pi, \pi]$. Before we make this statement more precise, note that in the limit of $L \rightarrow \infty$, $b(x, L, \chi)$ is a perfect step function which equals $1$ in the interval $[-\chi, \chi]$ and zero otherwise. For a finite $L$, $b(\lambda, L, \chi)$ will then be a smoothed step function and therefore, ${g}(\lambda)$ will only approximate $f(\lambda)$ over $[-\chi, \chi]$. By bounding Eq. \eqref{eq:func_filter} one can see that
\begin{equation}
\label{eq:L_bound}
\begin{aligned}
|f(\lambda) - {g}(\lambda)|   &< \frac{1}{2}(erf[L(\lambda + \chi)] -1) \\
                                            &< \frac{1}{2}erfc[L(\chi - 1)] \\
                                            &< \frac{1}{2}e^{-L^{2}(\chi - 1)^2},
\end{aligned}
\end{equation}
for all $\lambda \in [-1, 1]$, where $erfc$ is the complementary error function. Thus, if we take
\begin{equation}
\label{eq:L_bound_final}
L_{\varepsilon} \geq \frac{1}{\chi - 1}\sqrt{\ln\left(\frac{3}{2\varepsilon}\right)},
\end{equation}
it is guaranteed that
\begin{equation}
\label{eq:func_filtered_bound}
|f(\lambda) - {g}(\lambda)| < \frac{\varepsilon}{3}.
\end{equation}
Moreover, because $g(\lambda)\rightarrow 0$ for $\lambda\rightarrow\pm\pi$, the periodic extension of $g$ with period $2\pi$ has no jump discontinuities.
%%%%%%%%%%%%%%%%%%%%%%%%%%%%%%%%%%%%%%%%%%%%%%%%%%%%%%%%%
\begin{figure}[t!]
\centering
\includegraphics[width=1\columnwidth]{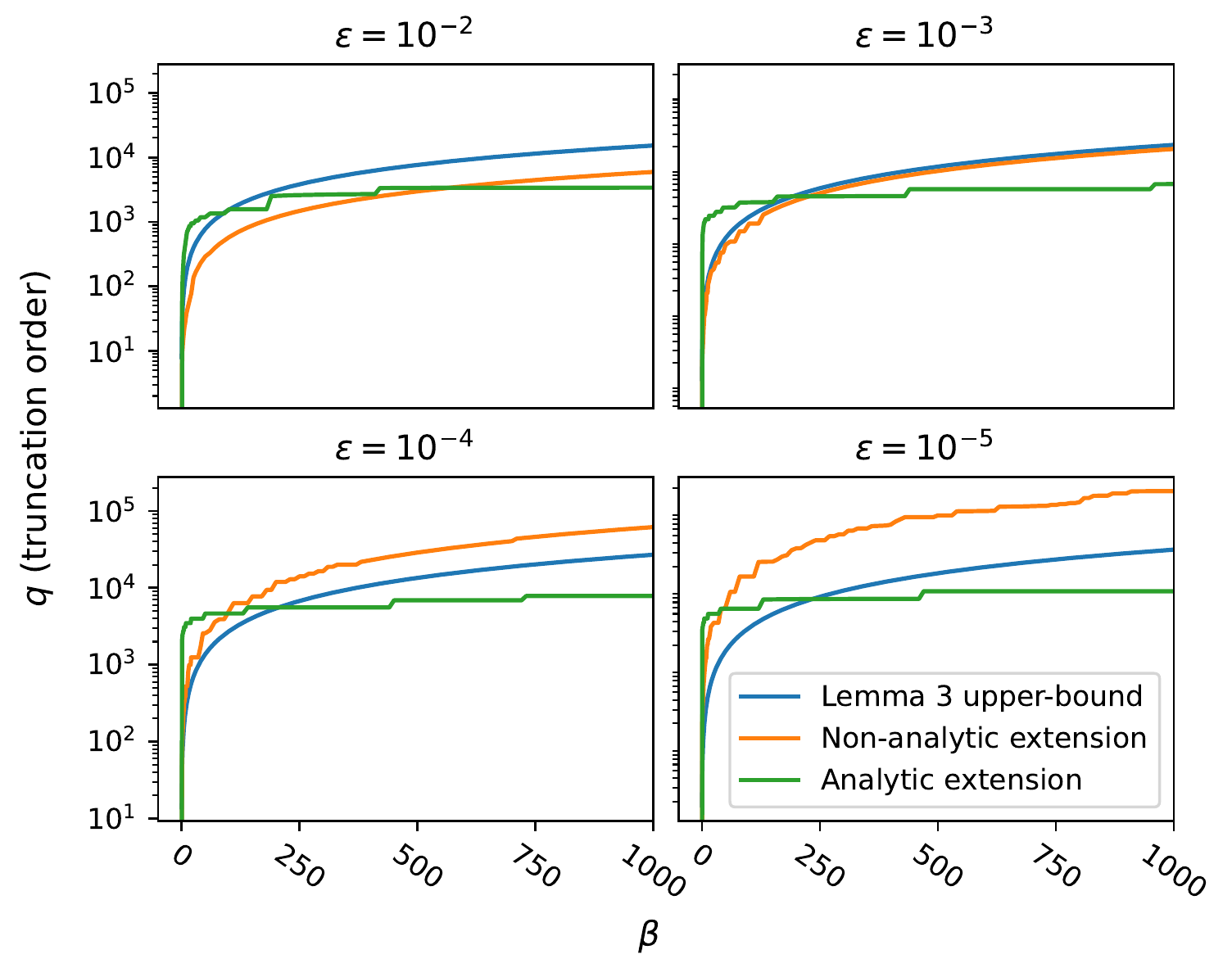}
\caption{\label{fig:mcomparisson} \textbf{Comparison among different Fourier approximations for the real exponential function.} On each panel, the truncation order (i.e. query complexity $q$) necessary to obtain an error bounded by $\varepsilon$ when approximating the function $f({\lambda}) = e^{-\beta(\lambda + 1)}$ is shown,  with $\beta\in\mathbb{R}$ the  inverse temperature, for four different values of $\varepsilon$. The blue line is the $q$ given by Eq. \eqref{eq:q_rt} with $\delta = 2\pi/(\sqrt{\beta + 4} + \sqrt{\beta})^{2}$, taken from \cite{silva2022fragmented}, chosen as to optimize the overall number of queries $Q_\Psi(f[H],\varepsilon,\alpha):=\frac{1}{p_\Psi(\alpha f[H])}q_f(\varepsilon)$. The orange line represents the truncation order of the conventional Fourier approximation of a non-analytic, continuous periodic extension of $f(\lambda)$ by a function $g(x)$ with $h_1$ and $h_2$ in Eq.\ \eqref{eq:g} are given by linear interpolating functions with the same slope.
%$g(x_{\lambda} + 2\pi\ell) = f(x_{\lambda})$, $\forall \ell \in \mathbb{Z}, x_{\lambda} \in [-x_{0}, x_{0}]$ and the end-point $(2\pi\ell + x_{0}, g(2\pi\ell + x_{0}))$ is connected to $(2\pi(\ell + 1) - x_{0}, g(2\pi(\ell + 1) - x_{0}))$ by a straight line for any given $\ell \in \mathbb{Z}$. 
One can see that $g(x)$ obtained in this way is continuous (if $f$ is continuous), but it displays discontinuities in its first derivative. The green line shows the corresponding truncation order for the  analytic extension function $\tilde{g}(x_{\lambda}, L_{\varepsilon}, \chi _{\varepsilon})$. Both in the analytic and non-analytic extension cases, $q$ is found by means of a binary search \cite{BSearchNote}, obtaining the conventional Fourier coefficients and calculating the error relatively to $f$ on each trial for $q$. Although Lemma \ref{tayfou} beats the analytic-extension construction for  small $\beta$, the latter asymptotically outperforms the former for high $\beta$. Also, the higher the precision, the higher the temperature at which this advantage over Lemma \ref{tayfou} is achieved. We verified this behavior with precision as high as $10^{-9}$.
}
 \end{figure}{}
%%%%%%%%%%%%%%%%%%%%%%%%%%%%%%%%%%%%%%%%%%%%%%%%%%%%%%%%%

From Eq.\ \eqref{eq:L_bound_final}, we see that for $\chi \rightarrow 1$, one would need $L \rightarrow \infty$ to satisfy the error bound in Eq. \eqref{eq:func_filtered_bound} -- in fact, in that case, with $\varepsilon = 0$. This happens because, although in the open interval $(-1, 1)$ error $\varepsilon = 0$ can be obtained, at the boundaries we have $b(\pm\chi, L, \chi) \approx 1/2$. To fix this, we  consider $\chi$ such that $[-1, 1] \in (-\chi, \chi)$. Also, because $g$ must be bounded by one -- up to given error -- in the closed interval $[-1, 1]$, we fix $\chi=\chi_\varepsilon$ through the equation
\begin{equation}
\label{eq:chi_definition}
\max_{\lambda \in [-\chi_{\varepsilon}, \chi_{\varepsilon}]}|f(\lambda)| = 1 + \frac{\varepsilon}{3}.
\end{equation}
Here we are assuming that the filter $b(\lambda, L, \chi)$ decays faster than $f$ for $|\lambda| > \chi$. One can see, with an analysis similar to \eqref{eq:L_bound}, that this would hold as long as $f$ does not grow faster than $L(|\lambda| - \chi)e^{L^{2}(|\lambda| - \chi)^2}$ for $|\lambda| > \chi$.
In the end, we implement the truncated standard Fourier expansion of the function $g(\lambda)=\tilde{g}(\lambda, L_{\varepsilon}, \chi _{\varepsilon})$ with the guarantee of exponential convergence. Say that $\tilde{g}_{q}$ is such an expansion with error $\varepsilon_0=\varepsilon/3$, i.e, $|\tilde{g}_{q}({\lambda}) - {g}({\lambda})| \leq \varepsilon/3$ for ${\lambda} \in [-1, 1]$. If we take $L_{\varepsilon}$ as in Eq. $\eqref{eq:L_bound_final}$ and $\chi _{\varepsilon}$ as in Eq. $\eqref{eq:chi_definition}$, we have
\begin{equation}
\label{eq:extension_approx}
|\tilde{g}_{q}({\lambda}) - f(\lambda)| < \varepsilon,
\end{equation}
and just an insignificant sub-normalization factor over the whole interval, i.e, $|\tilde{g}_{q}(\lambda)| \leq 1 + \varepsilon$, $\forall \lambda \in [-\pi, \pi]$.

The downside of this method is that we do not provide an algorithm to compute the coefficients of the series $\tilde{g}_{q}(\lambda)$, nor do we provide an analytical way to obtain the truncation order $q$ as a function of the desired error $\varepsilon$. At first sight this might seem to invalidate the whole construction. However, as was already mentioned, this method came about precisely because Lemma \ref{tayfou} becomes infeasible for small $\delta$.
%when applied to the function $f(\lambda) = e^{-\beta (\lambda + 1)}$. 
Even having to evaluate the Fourier coefficients explicitly from its integral definition, and having to brute-force search the truncation order $q$, this construction turned out to outperform the previous one in terms of classical computation time for the particular studied case of the exponential function. Since one of the most important metrics of oracle-based methods is the number of queries to the oracle, which here is  the truncation order of the series used to approximate the desired function, a natural comparison to make between this method and the one presented in Lemma \ref{tayfou} is  the truncation order necessary to achieve the same error guarantee. In Fig. \ref{fig:mcomparisson} we can see how the methods compare to each other in this aspect.  Note that for Lemma \ref{tayfou} we can only use the upper-bound since the construction of the series for large $\beta$ was too expensive to be carried out, let alone construct many of them in order to binary search \cite{BSearchNote} the minimum truncation order that would guarantee the desired error. In Fig.\ \ref{fig:mcomparisson} one can also see the comparison with a non-analytic extension as well.

\section{Proofs of Theorems}
\label{sec:proofs}

\subsection{Real-variable function design: Proof of Theorem \ref{lem:qsp2} and classical algorithm for pulse angles}\label{sec:proofRV}

Before proceeding with the proof of Theo.\ \ref{lem:qsp2}, we state and prove a lemma that provides an algorithm to find a unitary operator that has  the target Fourier series $\tilde{g}_q(x)$ as one of its entry. This also generalizes the results in Ref.\ \cite{Haah2019product} in the sense of removing the parity constraint required there.

\begin{lem}\label{lem:complementary}[Complementary Fourier series]
 Given  $\tilde{g}_q(x)=\sum_{m=-q/2}^{q/2} c_m\, e^{im x}$, $q\in\mathbb{N}$ even, satisfying $|\tilde{g}_q(x)|\leq 1$ for all $|x|\leq\pi$, there is a Fourier series $\tilde{h}_q(x)=\sum_{m=-q/2}^{q/2} b_m\, e^{im x}$ (with the same order $q/2$), such that 
 \begin{equation}\label{eq:RFourier1}
  U_{\tilde{g}_q\tilde{h}_q}(x)= \left(
  \begin{array}{cc}
  \tilde{g}_q(x)&\tilde{h}_q(x)\\
  -\tilde{h}_q^*(x)&\tilde{g}_q^*(x)
  \end{array}
  \right)
 \end{equation}
 is an $SU(2)$ operator. Moreover, the coefficients $\{b_m\}_m$ can be calculated classically in time $\mathcal{O}\left(\text{poly}(q)\right)$.
\end{lem}
\begin{proof}
In order to $U_{\tilde{g}_q\tilde{h}_q}(x)$ be unitary, its entries must satisfy 
\begin{equation}\label{eq:unitarity2}
  |\tilde{g}_q(x)|^2+|\tilde{h}_q(x)|^2=1\quad \forall x\in\mathbb{R}.
 \end{equation}
 Let us define the Laurent polynomial $G(z)=\sum_{k=-q}^{q} a_k\, z^k$ such that for $z\in U(1)$ we have $G(z=e^{ix})=1-|\tilde{g}_q(x)|^2$. The coefficients of $G(z)$ can be obtained as
 \begin{equation}
  a_k=\begin{cases}
          -\sum_{l=-q/2+k}^{q/2}c_l c_{l-k}^*,&\quad k>0\\
          1-\sum_{l=-q/2}^{q/2}|c_l|^2,&\quad k=0\\
          -\sum_{l=-q/2}^{q/2+k}c_l c_{l-k}^*,&\quad k<0
      \end{cases}.
 \end{equation}
 $G(z)$ is such that, if the aimed $\tilde{h}_q(x)$ does exist, then for $z\in U(1)$ we have $G(z=e^{ix})=|\tilde{h}_q(x)|^2$. Therefore, the goal is to identify $G(z)$ with the product of a function with its complex conjugate in the unit circle. 
  
 Finally, take $p(z)$ as the degree-$2q$ polynomial such that $G(z)=z^{-q}p(z)$. If $\mathcal{L}=\{r_k\}_{k\in[2q]}$ is the list  of all roots of $p(z)$ with their multiplicities ordered by increasing modulus, then we can express
 \begin{equation}\label{eq:pz}
 \begin{split}
  G(z)&=a_qz^{-q}\prod_{k=1}^{2q}(z-r_k)\\
   & =\left[a_q\prod_{k=1}^{q}r_k\right]
\left[\prod_{k=1}^{q}\left(\frac{1}{z}-\frac{1}{r_k}\right)\right]\left[\prod_{k=q+1}^{2q} (z-r_k)\right].
\end{split}
\end{equation}

Now, notice that $G(z)$ is a real polynomial and it is zero or positive for $z$ in the unity circle. In particular, the former property, when applied to the unity circle ($z^*=1/z$, $z\in U(1)$) gives that
\begin{equation}
 G^*(z)=\left[a_q^*\prod_{k=1}^{q}r_k^*\right]\left[\prod_{k=1}^{q}\left(z-\frac{1}{r_k^*}\right)\right]\left[\prod_{k=q+1}^{2q}\left(\frac{1}{z}-r_k^*\right)\right]
\end{equation}
must be equal to Eq.\ \eqref{eq:pz}. Comparing the two equations leads to the conclusion that for each root $r_k\in\mathcal{L}$, $1/r_k^*$ is also in $\mathcal{L}$. This also automatically ensures that the constant $\left[a_q\prod_{k=1}^{q}r_k\right]$ is a real number by noticing that $a_q\prod_{k=1}^{q}r_k/r_k^{*}$ is the constant term in $p(z)$ since the product contains all its roots. Whereas the constant term is $a_{-q}=c_{-q/2}c_{q/2}^*=a_q^*$. 

At last, we can use the positiveness of $G(z)$ in the unit circle:
\begin{equation}
 G(z)=\left[a_q\prod_{k=1}^{q}r_k\right]\left[\prod_{k=1}^{q}\left(z^*-\frac{1}{r_k}\right)\left({z}-\frac{1}{r_k^*}\right)\right]\geq 0,
\end{equation}
for $z\in U(1)$, 
to conclude that $\left[a_q\prod_{k=1}^{q}r_k\right]\geq 0$. Here we used that $z^*=1/z$ and assumed that the second half of $\mathcal{L}$ is composed by $1/r_k^*$ for each $r_k$ in the first half. 

Finally, still for $z\in U(1)$ we can write 
\begin{equation}
\begin{split}
 G(z)=& \sqrt{a_q\prod_{k=1}^{q}r_k} \left[ z^{-q/2}\prod_{k=1}^{q}\left({z}-\frac{1}{r_k^*}\right)\right]\\
  &\quad\sqrt{a_q\prod_{k=1}^{q}r_k}\left[z^{q/2}\prod_{k=1}^{q}\left(z^*-\frac{1}{r_k}\right)\right], 
 \end{split}
\end{equation}
and identify $\tilde{h}_q(x)$ as the polynomial in the first line for $z=e^{ix}$, such that $G(z=e^{ix})=\tilde{h}_q(x)\tilde{h}_q^*(x)$. The complexity of calculating the coefficients of $\tilde{h}_q(x)$ is basically given by the complexity of finding the roots of the polynomial $G(z)$ and thus is $\mathcal{O}(\text{poly}(q))$.
\end{proof}

Now that the operator $U_{\tilde{g}_q\tilde{h}_q}(x)$ has been built, we can prove Theorem\ \ref{lem:qsp2}. The main idea is to show that $U_{\tilde{g}_q\tilde{h}_q}(x)$ can be expressed as  a product of the basic qubit gates $R(x,\omega,\boldsymbol{\xi})$. 

\begin{proof}[Proof of Theo. \ref{lem:qsp2}] First of all, we prove that the matrix elements of $\mathcal{R}(x,\boldsymbol{\omega},\boldsymbol{\Phi})$ are indeed Fourier series in $x$ with the correct frequency values. The basic QSP gate $R(x,\omega,\boldsymbol{\xi})=e^{i\frac{\zeta+\eta}{2} Z}e^{-i\varphi Y}e^{i\frac{\zeta-\eta}{2} Z}e^{i\omega x Z}e^{-i\kappa Y}$ with $\boldsymbol{\xi}=\{\zeta,\eta,\varphi,\kappa\}$ can be conveniently represented as \cite{PerezSalinas2021}
\begin{equation}
 R(x,\omega,\boldsymbol{\xi})=\left(
 \begin{array}{cc}
   a_+e^{i\omega x} + a_-e^{-i\omega x}& b_+e^{i\omega x} + b_-e^{-i\omega x} \\
    -b_-^*e^{i\omega x} - b_+^*e^{-i\omega x}& a^*_-e^{i\omega x} + a^*_+e^{-i\omega x} \\
 \end{array}
\right),
\end{equation}
with 
\begin{equation}\label{eq:amais}
 \begin{split}
 a_+=\cos\varphi\,\cos\kappa\, e^{i\zeta}&\qquad a_-=-\sin\varphi\,\sin\kappa\, e^{i\eta}\\
 b_+=-\cos\varphi\,\sin\kappa \,e^{i\zeta}&\qquad b_-=-\sin\varphi\,\cos\kappa\, e^{i\eta}.
 \end{split}
\end{equation}

We first use mathematical induction to prove that, for any even $q$, $\mathcal{R}(x,\boldsymbol{\omega},\boldsymbol{\Phi})$ is a Fourier series with frequencies $\{-\frac{q}{2},-\frac{(q-1)}{2}, \cdots,0,\cdots,\frac{(q-1)}{2},\frac{q}{2}\}$.
Let us define $\mathcal{R}^{(m)}(x,\boldsymbol{\omega},\boldsymbol{\Phi})$ as the partial product up to the $m$-th term of $\mathcal{R}(x,\boldsymbol{\omega},\boldsymbol{\Phi})$.
 We are taking $\omega_0=0$, $\omega_k=(-1)^{-1}1/2$ for all $k\neq 0$ even. Therefore, the entries of $\mathcal{R}^{(0)}(x,\boldsymbol{\omega},\boldsymbol{\Phi})$ are  $0$-order Fourier series is obtained, i.e complex constants. Moreover, if none of the QSP parameters in Eq.\ \eqref{eq:amais} is zero, then  the matrix elements of the operator resulting from the  multiplication of  $R(x{,}\omega_k,{\boldsymbol{\xi}}_k)$ and $R(x{,}\omega_{k+1},\boldsymbol{\xi}_{k+1})$ for $k\geq1$ are  order-$1$ Fourier series with all the frequencies $\{-1,0,1\}$. In particular, $\bra{0}\mathcal{R}^{(2)}(x,\boldsymbol{\omega},\boldsymbol{\Phi})\ket{0}$ is an order-$1$ Fourier series. Proceeding with the induction,  we now assume that, for $m\geq4$ even, the components of $\mathcal{R}^{(m-2)}(x,\boldsymbol{\omega},\boldsymbol{\Phi})$ are order-$(\frac{m-2}{2})$ Fourier series with all terms with frequencies $\{-(\frac{m}{2}-1),\cdots,0,\cdots,(\frac{m}{2}-1)\}$. Multiplying the next two iterators  $R(x,\omega_{m-1},\boldsymbol{\xi}_{m-1})$ and $R(x{,}\omega_{m},\boldsymbol{\xi}_{m})$ will add or subtract $1$ to the frequencies, or keep the same frequency of the terms in $\mathcal{R}^{(m-2)}(x,\boldsymbol{\omega},\boldsymbol{\Phi})$. This implies that $\mathcal{R}^{(m)}(x,\boldsymbol{\omega},\boldsymbol{\Phi})$ is a also a Fourier series with the due frequencies. Therefore, we conclude that this is true for $\mathcal{R}(x,\boldsymbol{\omega},\boldsymbol{\Phi})$ with any $q$ even.
 
 We still need to prove that, we can always find QSP pulses $\boldsymbol{\Phi}$ such that the epper left-hand matrix entry of $\mathcal{R}(x,\boldsymbol{\omega},\boldsymbol{\Phi})$ gives the desired Fourier series $\tilde{g}_q(x)$. From Lem.\ \ref{lem:complementary} we know how to obtain the complementary Fourier series $\tilde{h}_q$. 
The right-hand side of Eq.\ \eqref{eq:unitarity2} is constant, implying that all the coefficients of oscillating  terms in the left-hand side must vanish. In particular, the highest-frequency term gives $c_{q/2} c_{-q/2}^*+d_{q/2} d_{-q/2}^*=0$. Our strategy is to successively multiply $U_{\tilde{g}_q\tilde{h}_q}$ by the inverse of the QSP fundamental gates, finding the pulse ${\boldsymbol{\xi}}_k$ that reduces the frequency of its components in each turn, until we obtain the identity operator. This will at the same time show that $U_{\tilde{g}_q\tilde{h}_q}=\mathcal{R}(x,\boldsymbol{\omega},\boldsymbol{\Phi})$ and find the desired $\boldsymbol{\Phi}$. We start by multiplying Eq.\ \eqref{eq:RFourier1} from the left by $R^{-1}(x{,}\omega_q,\boldsymbol{\xi}_q)$, obtaining 
\begin{equation}
\begin{split}
U_{\tilde{g}_{q-1}\tilde{h}_{q-1}}:=& R^{-1}(x{,}\omega_q,\boldsymbol{\xi}_q)U_{\tilde{g}_q\tilde{h}_q}\\
 =& \!\left(\!
  \begin{array}{cc}
  \tilde{g}_{q-1}(x)&\tilde{h}_{q-1}(x)\\
  -\tilde{h}_{q-1}^*(x)&\tilde{g}_{q-1}^*(x)
  \end{array}
  \!\right),
  \end{split}
\end{equation}
where 
\begin{equation}
\begin{split}
 \tilde{g}_{q-1}(x)=&\sum_{m=-q/2}^{q/2}\left(a^*_{q+}e^{-i\frac{x}{2} } 
 + a^*_{q-}e^{i\frac{x}{2} } \right)c_m e^{im x}\\  & +\left(b_{q+}e^{i\frac{x}{2} } + b_{q-}e^{-i\frac{x}{2}}\right)d^*_m e^{-im x}
\end{split}
\end{equation}
and 
\begin{equation}
 \begin{split}
  \tilde{h}_{q-1}(x)=&\sum_{m=-q/2}^{q/2}\left(a^*_{q+}e^{-i\frac{x}{2} } + a^*_{q-}e^{i\frac{x}{2} } \right)d_m e^{im x}\\
   &-\left(b_{q+}e^{i\frac{x}{2} } + b_{q-}e^{-i\frac{x}{2} }\right)c^*_m e^{-im x}.
 \end{split}
\end{equation}
The terms oscillating with frequencies $(\frac{q+1}{2})$ and $-(\frac{q+1}{2})$ must vanish, since they increase the frequencies present in $U_{\tilde{g}_{q}\tilde{h}_{q}}$. This leads to the two linear systems
\begin{equation}\label{eq:sistema1}
 \begin{split}
  a_{q-}^*c_{q/2}+b_{q+}d_{-q/2}^*=0\\
  a_{q-}^*d_{q/2}-b_{q+}c_{-q/2}^*=0
 \end{split}
\end{equation}
and 
 \begin{equation}\label{eq:sistema2}
 \begin{split}
  a_{q+}^*c_{-q/2}+b_{q-}d_{q/2}^*=0\\
  a_{q+}^*d_{-q/2}-b_{q-}c_{q/2}^*=0,
 \end{split}
\end{equation}
 that must be solved for the variables $\{a_{q+},a_{q-},b_{q+},b_{q-}\}$. This leads to the QSP angles of the $q-$th pulse via Eq.\ \eqref{eq:amais}. The unitarity condition in Eq.\ \eqref{eq:unitarity2} implies that  $c_{q/2} c_{-q/2}^*+d_{q/2} d_{-q/2}^*=0$. This guarantees that both linear systems \eqref{eq:sistema1} and \eqref{eq:sistema2} admit non-trivial solutions. Using Eq.\ \eqref{eq:amais}, it can be directly verified that both \eqref{eq:sistema1} and \eqref{eq:sistema2} lead to $\tan\varphi_q=-e^{i(\zeta_q + \eta_q)}\frac{d_{-q/2}^*}{c_{q/2}}$ and $\zeta_q + \eta_q=-\text{Arg}\left(\frac{d_{-q/2}^*}{c_{q/2}}\right)$. Equations \eqref{eq:sistema1} and \eqref{eq:sistema2} impose no restriction over $\kappa_q$ , which we fix as $\pi/4$. Also, one  rotation can be eliminated by choosing $\eta_q=\zeta_q$. 
 
 The $(q-1)$-th pulses are obtained in the exact same way from the highest-order coefficients of $\tilde{g}_{q-1}(x)$ and $\tilde{h}_{q-1}(x)$. Because the frequency of the $(q-1)$-th iterator is $\omega_{q-1}=-1/2$, the solution found for the angles is slightly different: $\tan\varphi_{q-1}=e^{i(\zeta_{q-1} + \eta_{q-1})}\frac{\tilde{c}_{(q/2}}{\tilde{d}_{-q/2}^*}$ and $\zeta_{q-1} + \eta_{q-1}=\text{Arg}\left(\frac{\tilde{c}_{q/2}}{\tilde{d}_{-q/2}^*}\right)$, where $\tilde{c}_{q/2}$ is the coefficient of $e^{i(\frac{q-1}{2}) x}$ in $\tilde{f}_{q-1}(x)$, and $\tilde{d}_{-q/2}$ accompanies $e^{-i(\frac{q-1}{2}) x}$ in $\tilde{g}_{q-1}(x)$. After the $(q-1)$-th inverse gate, we are left with an order-$(\frac{q-2}{2})$ Fourier series. We can repeat this process until we get an order-$0$ Fourier series, which determines the 
parameters of the last gate with frequency $\omega_0=0$ to be $\kappa_0=0$, $\zeta_0=\text{Arg}(\tilde{g}_0)$, $\eta_0=\text{Arg}(\tilde{h}_0)$, and $\tan \varphi_0=-e^{i(\zeta_0-\eta_0)}\frac{\tilde{h}_0}{\tilde{g}_0}$. To obtain the pulses, $\mathcal{O}(q^2)$ arithmetic operations are used.
\end{proof}

\subsection{Operator function design: proof of Theorem\ \ref{fourier}}\label{sec:opfunproof}

\begin{proof}[Proof of Theorem\ \ref{fourier}]
The proof consists of showing that the operator $\boldsymbol{V}_{\boldsymbol{\Phi}}=\left(
 \bar{V}_{\boldsymbol{\xi}_{q}}{V}_{\boldsymbol{\xi}_{q-1}}\cdots \bar{V}_{\boldsymbol{\xi}_{2}}{V}_{\boldsymbol{\xi}_{1}} \right)W_{\text{in}}$, with  $W_{\text{in}}$, $V_{\boldsymbol{\xi}}$, and $\bar{V}_{\boldsymbol{\xi}}$ given by Eqs. \eqref{eq:WinRT},  \eqref{eq:vphiRT}, and \eqref{eq:vphibarRT}, respectively,  can be written as 
\begin{equation}\label{eq:vPhiqsp2}
 V_{{\boldsymbol{\Phi}}}=\sum_{\lambda}\ket{\lambda}\bra{\lambda}\otimes \mathcal{R}\left({x_{\lambda}},\boldsymbol{\omega},{\boldsymbol{\Phi}}\right),
\end{equation}
with $x_\lambda=\lambda t$ and $\boldsymbol{\omega}=\{0,\frac{1}{2},-\frac{1}{2},\cdots,\frac{1}{2},-\frac{1}{2}\}$. The single-qubit operator $\mathcal{R}\left({x_{\lambda}},\boldsymbol{\omega},{\boldsymbol{\Phi}}\right)$ satisfies Theorem \ref{lem:qsp2}. Therefore, given a Fourier series $\tilde{g}_q(x)$,  with $|\tilde{g}_q(x)|\leq 1$ for all $x\in[-\pi,\pi]$, it is possible to efficiently find  $\boldsymbol{\Phi}$ from ${\mathbf{c}}=\{{c}_m\}_m$ such that 
\begin{equation}
 \bra{0}V_{{\boldsymbol{\Phi}}}\ket{0}=\sum_{\lambda}\ket{\lambda}\bra{\lambda} \bra{0}\mathcal{R}\left({x_{\lambda}},\boldsymbol{\omega},{\boldsymbol{\Phi}}\right)\ket{0}
\end{equation}
reduces to Eq.\ \eqref{eq:rtrealfunction} with $\bra{0}\mathcal{R}\left({x_{\lambda}},\boldsymbol{\omega},{\boldsymbol{\Phi}}\right)\ket{0}=\tilde{g}_q(x_\lambda)$.

First of all, notice that Eq.\ \eqref{eq:WinRT} can be written, using the definition of $R(x,\omega,\boldsymbol{\xi})$, as
\begin{equation}\label{eq:WinRT3}
 W_{\text{in}}=\sum_\lambda \ketbra{\lambda}{\lambda}\otimes R(x_\lambda,\omega=0,\boldsymbol{\xi}_0).
\end{equation}

 Using the spectral decomposition of the Hamiltonian $H$, the oracle $O$ can be written as
 \begin{equation}\label{eq:epecOracle}
  O=\sum_\lambda \ketbra{\lambda}{\lambda}\otimes\left(\ketbra{0}{0}+e^{-ix_\lambda }\ketbra{1}{1}\right), 
 \end{equation}
or, alternatively, $O=\sum_\lambda e^{-i\frac{x_\lambda}{2}} \ketbra{\lambda}{\lambda}\otimes e^{i\frac{x_\lambda}{2}Z}$, which directly implies that Eq.\ \eqref{eq:vphiRT} can be expressed as 
\begin{equation}\label{eq:vphiRT3} 
V_{\boldsymbol{\xi}_k}=\sum_{\lambda} \ket{\lambda}\bra{\lambda}\otimes {R}\left(x_{\lambda},\omega=1/2,{\boldsymbol{\xi}_k}\right).
\end{equation}
Analogously, from Eq.\ \eqref{eq:vphibarRT} it holds 
\begin{equation}\label{eq:vphibarRT3}
\bar{V}_{\boldsymbol{\xi}_k}=\sum_{\lambda} \ket{\lambda}\bra{\lambda}\otimes {R}\left(x_{\lambda},\omega=-1/2,{\boldsymbol{\xi}_k}\right).
\end{equation}

Because of the orthogonality between distinct  eigenvectors of $H$, Eq.\ \eqref{eq:vPhiqsp2} follows straightforwardly from the product of Eq.\ \eqref{eq:WinRT3} and the alternated iteration of the operators in Eqs. \eqref{eq:vphiRT3} and \eqref{eq:vphibarRT3}, as desired.
\end{proof}

Finally, let us make  a last comment about the mapping of Hamiltonian eigenvalues to the convergence interval of the target Fourier series. Notice that taking $t=x_0$ maps any $\lambda\in[-1,1]$ to the convergence interval $[-x_0,x_0]$ of the Fourier series. However, we can map any interval of eigenvalues $\left[\lambda_{+},\lambda_{-}\right]$ to $[-x_0,x_0]$. This is done by replacing the oracle $O$ with the iterate 
  $V_0=\mathds{1}\otimes \ketbra{0}{0}+e^{-i\Lambda}e^{-iHt}\otimes\ketbra{1}{1}$, which can be obtained from $O$ using one query. We can write  $V_0=\sum_\lambda e^{-i\frac{x_\lambda}{2}} \ketbra{\lambda}{\lambda}\otimes e^{i\frac{x_\lambda}{2}Z}$, now with $x_\lambda=\lambda t +\Lambda$, and again a function $\alpha f(\lambda)=g(x_\lambda)$ can be $\varepsilon$-block-encoded into $V_{\boldsymbol{\Phi}}$. The eigenvalues interval $\left[\lambda_{+},\lambda_{-}\right]$ is then mapped into the convergence interval if we take $t=x_0/\Delta\lambda$ and $\Lambda=-x_0\frac{\bar{\lambda}}{\Delta\lambda}$, with $\Delta\lambda=(\lambda_+-\lambda_-)/2$ and $\bar{\lambda}=(\lambda_++\lambda_-)/2$.

\section{Discussion}
\label{sec:Discussions}

We presented a novel quantum signal processing varian{t able to produce any  Fourier series} of a Hermitian operator $H$. The algorithm assumes access to {a Hamiltonian oracle given by a single-qubit controlled version} of the unitary operator $e^{-iHt}$. 
%that implements the unitary operator $e^{-iHt}$ controlled by a single-qubit ancilla.
Remarkably, this is the only ancilla required throughout the algorithm. The evolution time $t$ {in} the oracle is fixed {at a tunable value}. By interspersing {oracle calls with parameterized single-qubit rotations on the ancilla}, the target operato{r f}unction is {encoded into} a block of the resulting {joint-system} unitary transformation. The operator function is then {physically realized on the system by a measurement postselection on} the ancilla.

More technically, the problem of {synthesizing a Fourier series of an operator is reduced to applying a single-qubit pulse sequence that realizes the corresponding real-variable series as a matrix element of an $SU(2)$ operator. We provided an explicit efficient classical algorithm for determining the parameters in such sequence from the target Fourier-series coefficients, and rigorously proved that any (normalized) finite Fourier series can be implemented}. Hence, our method is able to implement arbitrarily good approximations to any Hamiltonian function with a finite number of jump discontinuities. The {tunable} evolution time {in the oracle gives us f}reedom to map the range of eigenvalues of $H$ to {an effective interval of interest shorter than the period of the approximation series, which allows us to control the convergence in a practical way and avoid the Gibbs phenomenon at the boundaries}. 

{Another important task we considered is how to} find the best Fourier approximation to be used{, given a fixed target function. We presented two such classical sub-routines for such approximations. The first one is based on Lemma 37 of Ref.\ \cite{vanApeldoorn2020quantumsdpsolvers} and comes guaranteed error bounds. The second one is obtained from an analytic periodic extension of the target function to Fourier-approximate that has guaranteed asymptotic convergence but no closed-form expression for the approximation error. However, while the former is so computationally intensive that it becomes in practice prohibitive already for moderate applications, the former is consistently observed to be numerically very efficient}. 
%We compared  the resulting query complexity of the two methods for the case of the exponential function $e^{-\beta H}$. We obtained that the bounded-error one is better only for small $\beta$. Nevertheless,  we compared the truncation order of the actual Fourier series obtained from the analytic extension with the theoretical bound of \cite{vanApeldoorn2020quantumsdpsolvers}, since obtaining the series for the latter proved impractical.  
Moreover, the {second method circumvents the sub-normalization required by the first method, thus leading also to significantly higher post-selection probabilities}.  

{Our findings provide a versatile tool-box for generic operator-function synthesis, relevant to a plethora of modern quantum algorithms, with the advantages of reduced number of ancillary qubits, compatibility with Trotterised Hamiltonian simulations schemes, and relevance both for digital as well as hybrid digital-analog quantum platforms.}

%%%%%%%%%%%%%%%%%%%%%%%%%%%%%%%%%%%%%%%%%%%%%%%%%%%%%%%%
%%%%%%%%%%%%%%%%%%%%%%%%%%%%%%%%%%%%%%%%%%%%%%%%%%%%%%%%
\begin{acknowledgments}
We thank Adrian Perez Salinas and Jose Ignacio Latorre for discussions. We acknowledge financial support from the Serrapilheira Institute (grant number Serra-1709-17173), and the Brazilian agencies CNPq (PQ grant No. 305420/2018-6) and FAPERJ (JCN E-26/202.701/2018).
\end{acknowledgments}
%%%%%%%%%%%%%%%%%%%%%%%%%%%%%%%%%%%%%%%%%%%%%%%%%%%%%%%%
%%%%%%%%%%%%%%%%%%%%%%%%%%%%%%%%%%%%%%%%%%%%%%%%%%%%%%%%

%\bibliographystyle{plainnat}
\bibliography{notas}

\begin{thebibliography}{30}
\bibitem [{BSe()}]{BSearchNote}%
  \BibitemOpen
  \href@noop {} {}\bibinfo {note} {Binary search is an algorithm for searching
  an element in a sorted array by repeatedly dividing the search space by half
  in each iteration. The basic step of the algorithm consists in picking the
  element halfway the array and testing the target of the search against it.
  Assuming ascending ordering, if the target is bigger than the mid-element,
  owing to the fact that the array is sorted, one can ignore the lower-half and
  repeat the process. If the target is smaller than the said element, one
  ignores the upper-half and again, repeats the process. Convergence happens
  after at most $\lceil \log_{2}(n) \rceil$ steps, where n is the size of the
  array. Note that the naive one-by-one comparison used in the linear search
  requires $n$ steps in the worst-case scenario.}\BibitemShut {Stop}%
\end{thebibliography}

\appendix

\section{Another algorithm for operator function design using Fourier series}
\label{app:proofFourier} 

In this appendix we show that it is possible to implement  Fourier series of Hermitian operators using the standard QSP pulses. Nevertheless, this approach has achievability guaranteed only for real Fourier series. Moreover, it requires one extra qubit ancilla and extra subnormalization of $1/2$.

\subsection{Real-function design with single-qubit rotations}
\label{sec:pulses}

We start by reviewing the design of functions of one real variable with single-qubit pulses presented in Ref.\ \cite{Low2016PRX}.

The basic  single qubit rotation in this method is given by
$R(x,\phi)=e^{i x {X}}e^{i\phi{Z}}$, where ${X}$
and ${Z}$ are the first and third Pauli matrices, respectively,
and $\phi\in[0,2\pi]$. The angle $x\in[-\pi,\pi]$
is the signal to be processed and the rotation $e^{i x {X}}$
is now the iterate. For a given even number $q$, the sequence of rotations $\mathcal{R}\left(x,{\boldsymbol{\phi}}\right)=e^{i\phi_{q+1}{Z}}\prod_{k=1}^{q/2}R(- x,\phi_{2k})R( x,\phi_{2k-1})$
 results in an operator with matrix representation \cite{Low2016PRX}
\begin{equation}
\mathcal{R}\left(x,{\boldsymbol{\phi}}\right)=\left(\begin{array}{cc}
B(\cos x) & i\,\sin x\,D(\cos x)\\
i\,\sin x\,D^*(\cos x) & B^{*}(\cos x)
\end{array}\right),\label{eq:poly}
\end{equation}
whose entries contain polynomials $B$ and $D$ in $\cos x$ with complex coefficients
determined by the sequence of rotation angles ${\boldsymbol{\phi}}=\left(\phi_{1},\cdots,\phi_{q+1}\right)\in\mathbb{R}^{q+1}$. 

Given target real functions $\mathscr{B}(x )$ and $\mathscr{D}( x)$ satisfying 
\begin{equation}\label{achievabilityCond1}
 \mathscr{B}^{2}(x)+\sin^2 x \,\mathscr{D}^{2}( x)\leq1
\end{equation}
for all $x$ and having the form 
\begin{equation} \label{achievabilityCond2}
\begin{split}
  \mathscr{B}(x)& =\sum_{k=0}^{q/2}b_k\cos(2k x)\\
  \sin x\,\mathscr{D}( x)& =\sum_{k=1}^{q/2} d_{k}\sin\left(2k x\right),
  \end{split}
\end{equation}
with  $b_k$ and $d_k$  arbitrary real coefficients, then there is a pulse sequence ${\boldsymbol{\phi}}$ that generates $B(\cos x)$ and $D(\cos x)$ with $\mathscr{B}(x )$ and $\mathscr{D}(x)$ as either their real or imaginary parts, respectively.

Given that the desired polynomials are achievable by QSP, th{e r}otation angles ${\boldsymbol{\phi}}$ can be computed classically in time
$\mathcal{O}\left(\text{poly}(q)\right)$ \citep{Low2016PRX,Haah2019product,chao2020finding,dong2020efficient}.

\subsection{Multiqubit operator-function design}

We again assume that the Hermitian operator $H$ is encoded in the real-time evolution oracle of Def. \ref{def:real_t_or} and show how to apply the QSP pulses of Sec.\ \ref{sec:pulses} in order to implement an operator  Fourier series. As the pulses already produce sine and cosine series, we use an extra ancilla denoted by $\mathcal{A}_P$ to combine the two opposite parity series into an arbitrary real Fourier series.   Denoting by $\mathcal{A}_O$ the oracle ancilla, the total Hilbert space for the ancillas is now $\mathbb{H}_{\mathcal{A}}=\mathbb{H}_{\mathcal{A}_O}\otimes \mathbb{H}_{\mathcal{A}_P}$.

 Defining $O'=(\mathds{1}\otimes M) O (\mathds{1}\otimes M)$, the basic QSP blocks $\mathcal{V}_{\phi^{(j)}}=O'(\mathds{1}\otimes e^{i\phi^{(j)} {Z}})$ and $\bar{\mathcal{V}}_{\phi^{(j)}}=O'^\dagger(\mathds{1}\otimes e^{i\phi^{(j)} {Z}})$ can be obtained from one call to the oracle $O$ or its inverse. Here the qubit rotations are applied on the oracle ancilla. Depending on the preparation and post-selection of this ancilla, Eq. \eqref{achievabilityCond2}  allows for a sine or a cosine series on $x_\lambda=\lambda t$  by iterating $\mathcal{V}_{\phi^{(j)}}$. It means that only series with well defined parity can be implemented. The addition of the second qubit ancilla makes possible to implement a series with undefined parity. With $\boldsymbol{V}_{{\boldsymbol{\phi}}}=W_{\text{out}}\bar{V}_{\phi_q}{V}_{\phi_{q-1}}\cdots\bar{V}_{\phi_2}{V}_{\phi_{1}}W_{\text{in}}$, for this algorithm we take input and output ancillary unitaries as
\begin{equation}\label{eq:WinRT2}
W_{\text{in}}=\mathds{1}\otimes M\otimes M,
\end{equation}
where $M$ is a qubit Hadamard gate, and 
\begin{equation}\label{eq:WoutRT2}
\begin{split}
W_{\text{out}}=W_{\text{in}}&\left[\mathds{1}\otimes\left( e^{i\phi_{q+1}^{(c)} {Z}}\otimes \ket{0}\bra{0} \right.\right. \\
&\qquad\quad+\left.\left.( {Z} e^{i\phi_{q+1}^{(s)} {Z}})\otimes \ket{1}\bra{1}\right)\right],
\end{split}
\end{equation}
and the basic QSP blocks are defined as
 \begin{equation}\label{eq:vphiRT2}
  {V}_{\phi_k}=\mathcal{V}_{\phi_k^{(c)}}\otimes\ket{0}\bra{0}+\mathcal{V}_{\phi_k^{(s)}}\otimes \ket{1}\bra{1},
 \end{equation}
 \begin{equation}\label{eq:vphibarRT2}
  \bar{V}_{\phi_k}=\bar{\mathcal{V}}_{\phi_k^{(c)}}\otimes\ket{0}\bra{0}+\bar{\mathcal{V}}_{\phi_k^{(s)}}\otimes \ket{1}\bra{1},
 \end{equation}
with $\phi_k=(\phi_k^{(c)},\phi_k^{(s)})$.  The following lemma summarizes the method:

\begin{lem}\label{fourier2} (Fourier series from real-time evolution oracles - second approach) Let $\tilde{g}_q:[-\pi,\pi]\rightarrow\mathbb{R}$ be  the Fourier series $\tilde{g}_q(x)=b_0+\sum_{k=1}^{q/2} \left(b_k \cos k x +d_k\sin k x\right)$, with $|\tilde{g}_q(x)|\leq1$ for all $x \in [-\pi, \pi]$. Then there is a pulse sequence ${\boldsymbol{\phi}}=( \phi_1, \cdots,\phi_{q+1} )\in\mathbb{R}^{2q+2}$, with $\phi_k=({\phi}_k^{(c)},{\phi}_k^{(s)})$, such that the operator $\boldsymbol{V}_{{\boldsymbol{\phi}}}$ on $\mathbb{H}_{\mathcal{SA}}$  with $W_{\text{in}}$, $W_{\text{out}}$, ${V}_{\phi_k}$ and $\bar{V}_{\phi_k}$ given by Eqs. \eqref{eq:WinRT2}, \eqref{eq:WoutRT2}, \eqref{eq:vphiRT2}, and \eqref{eq:vphibarRT2}, respectively, is a perfect block-encoding of $\tilde{f}_q[H]{:}=\frac{1}{2} \tilde{g}_q[Ht]$, i.{e.} 
\begin{equation}
\bra{0}\bra{0}\boldsymbol{V}_{{\boldsymbol{\phi}}}\ket{0}\ket{0}=\sum_{\lambda}\tilde{f}_q(\lambda)\,\ket{\lambda}\bra{\lambda}.\label{eq:rtrealfunction2}
\end{equation}
Moreover, the pulse sequence can be obtained classically in time $\mathcal{O}(\text{poly}(q))$.
\end{lem}

\begin{proof}
 Consider the operator $\boldsymbol{\mathcal{V}}_{\boldsymbol{\phi}^{(j)}}=(\mathds{1}\otimes e^{i\phi_{q+1}^{(j)} {Z}})\bar{\mathcal{V}}_{\phi_q^{(j)}}{\mathcal{V}}_{\phi_{q-1}^{(j)}}\cdots\bar{\mathcal{V}}_{\phi_2^{(j)}}{\mathcal{V}}_{\phi_{1}^{(j)}}$ on $\mathbb{H}_{\mathcal{SA}_O}$, where $\boldsymbol{\phi}^{(j)}\in\mathbb{R}^{q+1}$ is the set of angles $\{\phi^{(j)}_{q+1},\cdots,\phi^{(j)}_{1}\}$, with $j=c$ or $j=s$. Using  Eq.\ \eqref{eq:epecOracle}, it is straightforward to see that 
\begin{equation}
\boldsymbol{\mathcal{V}}_{{\boldsymbol{\phi}^{(j)}}}=\sum_{\lambda}\ket{\lambda}\bra{\lambda}\otimes\mathcal{R}\left(\frac{x_{\lambda}}{2},{\boldsymbol{\phi}^{(j)}}\right),
\end{equation}
with the  concatenated qubit-rotations operator $\mathcal{R}\left(\frac{x_{\lambda}}{2},{\boldsymbol{\phi}}\right)$
given by Eq. \eqref{eq:poly}. Considering $\mathcal{A}_{O}$ to be initialized and
afterwards projected on state $\ket{+}$, the resulting operator on $\mathbb{H}_{\mathcal{S}}$
is obtained as 
\begin{equation}
\bra{+}\boldsymbol{\mathcal{V}}_{{\boldsymbol{\phi}^{(j)}}}\ket{+}=\sum_{\lambda}\left(\mathscr{\tilde{B}}(x_{\lambda})+i\mathscr{\tilde{D}}(x_{\lambda})\right)\ket{\lambda}\bra{\lambda}.\label{eq:cos}
\end{equation}
with $\mathscr{\tilde{B}}(x)=\text{Re}[B(\cos x)]$ and $\mathscr{\tilde{D}}(x)=\text{Re}[\sin x\,D(\cos x)]$.
By setting $\mathscr{\tilde{D}}(x)=0$, a real cosine series
on $x$ is obtained. A real sine series can also be implemented
through 
\begin{equation}
\bra{+}\left(\mathds{1}\otimes Z\right)\boldsymbol{\mathcal{V}}_{{\boldsymbol{\phi}^{(j)}}}\ket{+}=\sum_{\lambda}\left(i\mathscr{\bar{B}}(x_{\lambda})+\mathscr{\bar{D}}(x_{\lambda})\right)\ket{\lambda}\bra{\lambda}\label{eq:sin}
\end{equation}
by setting $\mathscr{\bar{B}}(x)=0$. Here we defined $\mathscr{\bar{B}}(x)=\text{Im}[B(\cos x)]$
and $\mathscr{\bar{D}}(x)=-\text{Im}[\sin x\,D(\cos x)]$.
Moreover, adding an extra control qubit $\mathcal{A}_{P}$ allows to perform a combination
of them by applying $\tilde{\boldsymbol{V}} _{\boldsymbol{\phi}}=\boldsymbol{\mathcal{V}}_{{\boldsymbol{\phi}^{(c)}}}\otimes \ket{0}\bra{0}_{\mathcal{A}_{P}}+\left[\left(\mathds{1}\otimes Z\right)\boldsymbol{\mathcal{V}}_{{\boldsymbol{\phi}^{(s)}}}\right]\otimes \ket{1}\bra{1}_{\mathcal{A}_{P}}$
since 
\begin{equation}
_{\mathcal{A}_{P}}\bra{+}\tilde{\boldsymbol{V}} _{\boldsymbol{\phi}}\ket{+}_{\mathcal{A}_{P}}=\frac{1}{2}\left(\boldsymbol{\mathcal{V}}_{{\boldsymbol{\phi}^{(c)}}}+\left(\mathds{1}\otimes Z\right)\boldsymbol{\mathcal{V}}_{{\boldsymbol{\phi}^{(s)}}}\right).\label{eq:four}
\end{equation}
In order to meet the achievability conditions, both cosine and sine
series must be normalized according to Eq.\ \eqref{achievabilityCond1}. 
Since  $\left|\tilde{g}_{q}(x)\right|<1$ for all $x\in[-\pi,\pi]$ then 
and
\begin{equation}\label{sn}
\begin{split}
\left|\sum_{k=0}^{q/2}b_k\,\cos k x\right|=&\left|\frac{\tilde{g}_{q}(x)+\tilde{g}_{q}(-x)}{2}\right| \leq1\\
 \left|\sum_{k=1}^{q/2}d_{k}\,\sin k x\right|=&\left|\frac{\tilde{g}_{q}(x)-\tilde{g}_{q}(-x)}{2}\right|\leq1
\end{split}
\end{equation}
such that the achievability
conditions of QSP are satisfied by the cosine and sine series individually. Therefore, there are two sets of angles ${\boldsymbol{\phi}}^{(c)}$ and
${\boldsymbol{\phi}}^{(s)}$ that realize the sine and the cosine series through
the QSP operators \eqref{eq:cos} and \eqref{eq:sin}, respectively.
Notice that setting $\mathscr{\tilde{D}}(\theta)=0$ and $\mathscr{\bar{B}}(\theta)=0$
does not compromise the achievability conditions to be fulfilled. The final result in Eq.\ \eqref{eq:rtrealfunction2} comes from the observation that $\boldsymbol{V}_{{\boldsymbol{\phi}}}=W_{\text{in}}\tilde{\boldsymbol{V}}_{{\boldsymbol{\phi}}}W_{\text{in}}$. The ancillas unitary $W_{\text{in}}$ serves only to change the ancillas projection to the computational basis.

\end{proof}

\end{document}